\documentclass[12pt]{elsart}
\usepackage{amssymb}
\usepackage{amsmath}
\usepackage{graphicx}
\usepackage{enumerate}

\newenvironment{proof}{ {\em Proof.} \,}{\qed}

\newcommand{\pd}{\partial}
\newcommand{\no}{\nonumber}

\begin{document}

\begin{frontmatter}

\title{Initial Value Problem of the Whitham Equations for the Camassa-Holm Equation}
\author[tg]{Tamara Grava}, 
\ead{grava@sissa.it}
\author[osu]{V. U. Pierce}, and 
\ead{virgilpierce@gmail.com}
\author[osu]{Fei-Ran Tian}
\ead{tian@math.ohio-state.edu}
\address[tg]{SISSA, Via Beirut 2-4, 34014 Trieste, Italy}
\address[osu]{Department of Mathematics, Ohio State University, 231 W. 18th Avenue, Columbus, OH 43210}


\begin{abstract}
We study the Whitham equations for the
Camassa-Holm equation. The equations are neither strictly hyperbolic nor
genuinely nonlinear. We are interested in the initial value problem of the Whitham equations.
When the initial values are given by a step function, the Whitham solution is self-similar.
When the initial values are given by a smooth function, the Whitham solution exists within a cusp
in the $x$-$t$ plane. On the boundary of the cusp, the Whitham solution matches the Burgers
solution, which exists outside the cusp.
\end{abstract}


\end{frontmatter}

\section{Introduction}

The Camassa-Holm equation
\begin{equation} \label{CH}
u_t + (3 u +2 \nu) u_x - \epsilon^2 (u_{xxt} + 2 u_x u_{xx} + u u_{xxx}) =  0 \ ,  \quad u(x, 0; \epsilon) = u_0(x)
\end{equation}
describes waves in shallow water when surface tension is present \cite{cam}. Here, $\nu$ is a constant
parameter. The solution of the initial value problem (\ref{CH}) will develop 
singularities in a finite time if and only if some portion of the positive part of
the initial ``momentum'' density $u_0(x) - \epsilon^2 u_0''(x) + \nu$ lies 
to the left of some portion of 
its negative part \cite{mck}. In particular, a unique global solution is guarenteed if 
$u_0(x) - \epsilon^2 u_0''(x) + \nu$ does not change its sign. 
These are the non-breaking 
initial data that we are 
interested in throughout this paper.

Although the zero dispersion limit of equation (\ref{CH}) has not been established,
some of its modulation equations (i.e., Whitham equations) have been derived.
The zero phase Whitham equation is
\begin{equation} \label{we0}
u_t + (3u + 2 \nu) u_x = 0 \ , \quad u(x,0) = u_0(x) \ ,
\end{equation}
which can be obtained from (\ref{CH}) by formally setting $\epsilon = 0$.

The single phase Whitham equations have been found in \cite{abe} and they can be
written in the Riemann invariant
form
\begin{equation}
\label{CHW}
u_{ix} + \lambda_i(u_1,u_2,u_3) u_{ix} = 0  \quad  \mbox{for} - \nu < u_3 < u_2 < u_1 \ ,
\end{equation}
where
\begin{equation}
\label{I}
\lambda_i(u_1, u_2, u_3) = u_1 + u_2 + u_3 + 2 \nu  - {I \over \partial_{u_i} I} \ ,
\end{equation}
and
$$I(u_1, u_2, u_3) = \int_{u_3}^{u_2} {\eta + \nu \over \sqrt{(\eta + \nu)(u_1 - \eta)(u_2 - \eta)(\eta - u_3)}} \ d \eta \ .$$
The constraint $ - \nu < u_3 < u_2 < u_1$ is consistent with the non-breaking initial 
data mentioned in the first paragraph of this section. 
The integral $I$ can be rewritten as a contour integral. Hence, it satisfies the Euler-Poisson-Darboux equations
\begin{equation}
\label{EPD}
2(u_i - u_j) {\partial^2 I \over \partial u_i \partial u_j} = {\partial I \over \partial u_i} - {\partial I \over \partial u_j} \ ,
\quad i, j = 1, 2, 3
\end{equation}
since the integrand does so for each $\eta \neq u_i$.  The Hamiltonian structure of
the single phase  Whitham equations for the Camassa-Holm equation was also obtained
in \cite{abe} in terms of Abelian integrals. The higher phase Whitham equations
can also
be derived using this structure.

In this paper we will study the
evolution of the Whitham solution from the zero phase to the single
phase.

This problem is similar to that of the zero dispersion limit of the KdV equation \cite{lax, lax2, ven}
\begin{equation}
\label{KdV}
u_t + 6 u u_x + \epsilon^2 u_{xxx} = 0 \ , \quad u(x,0; \epsilon) = u_0(x) \ .
\end{equation}
It is known that the zero phase Whitham equation for the KdV equation is
\begin{equation}
\label{Burgers}
u_t + 6 u u_x = 0 \ ,
\end{equation}
which is equivalent to (\ref{we0}) for the Camassa-Holm equation. The single phase
Whitham equations for the KdV equation are \cite{dub, fla, whi}
\begin{equation}
\label{KdVW}
u_{ix} + \mu_i(u_1,u_2,u_3) u_{ix} = 0 \ ,  \quad  \mbox{for} \quad
 u_3 < u_2 < u_1 \ ,
\end{equation}
where
$$\mu_i(u_1, u_2, u_3) = 2[u_1 + u_2 + u_3  - {\tilde{I} \over \partial_{u_i}
\tilde{I}}] \ ,$$
and
$$\tilde{I}(u_1, u_2, u_3) = \int_{u_3}^{u_2} {1 \over \sqrt{(u_1 - \eta)(u_
2 - \eta)(\eta - u_3)}} \ d \eta \ .$$
These equations are also similar to (\ref{CHW}) and (\ref{I}) for the
Camassa-Holm equation.

In the KdV case, the evolution from the zero phase to the single phase has
been studied in \cite{Tian1}. There, the Euler-Poisson-Darboux equations (\ref{EPD})
have played an important role. The same equations have also played a crucial role in
the study of the transition from the single phase to the double phase in \cite{GT}.

Although both the Camassa-Holm equation (\ref{CH}) and the KdV equation (\ref{KdV}) are
dispersive approximations to the Burgers equation (\ref{we0}) or (\ref{Burgers}), there are significant differences in the limiting dynamics.  The biggest difference is that the Whitham equations (\ref{CHW}) for the former equation are non-strictly hyperbolic (cf.  (\ref{ws1})) while the Whitham equations (\ref{KdVW}) for the
latter equation are strictly hyperbolic \cite{lev}. Non-Strictly hyperbolic Whitham equations have
also been found in the higher order KdV flows \cite{PT, PT2, PT3} and
the higher order defocusing NLS flows \cite{kod}. Self-similar solutions of
these Whitham equations have been constructed. They are remarkably different from
the self-similar solutions of the KdV-Whitham equations \cite{PT, PT2} or the
NLS-Whitham equations \cite{kod}, both of which are strictly hyperbolic.

In this paper, we will modify the method of paper \cite{Tian1} so that it can be used to
solve the non-strictly hyperbolic Whitham equations (\ref{CHW}) when the initial
function is a
smooth function. We will then study the evolution
from the zero phase to the single phase for smooth initial data. When the initial
function is a step-like function, we will use the method of paper \cite{kod, PT, PT2} to study the same
evolution.

The organization of the paper is as follows. In Section 2, we will introduce an initial
value problem which describes the evolution of phases. We will also discuss how to use the
hodograph transform to solve non-strictly hyperbolic Whitham equations.
In Section 3, we will study the properties of the eigenspeeds of the single phase
Whitham equations. We will study the initial value problem when the initial function
$u_0(x)$ is a step-like function in Section 4 and when $u_0(x)$ is a smooth
decreasing function in Section 5.

\section{An Initial Value Problem}

We describe the initial value problem for the Burgers equation (\ref{we0}) and
Whitham equations
(\ref{CHW}) as follows (see Figure 1.). Consider a horizontal motion
of the initial
curve $u = u_0(x)$.  At the beginning,
the curve evolves according to the Burgers equation (\ref{we0}).
The Burgers solution
breaks down in a finite time. Immediately
after the breaking,  the curve develops three branches. Denote
these three branches by $u_1$, $u_2$, and $u_3$.
Their motion is governed by the Whitham equations (\ref{CHW}).
As time goes on, the Whitham solution may develop singularities
and more branches are created. However, our focus is on the
evolution of the solution of the Whitham equations from the  one branch regime to the
three branch regime.

\begin{figure}[h]
\begin{center}
\includegraphics[width=12cm]{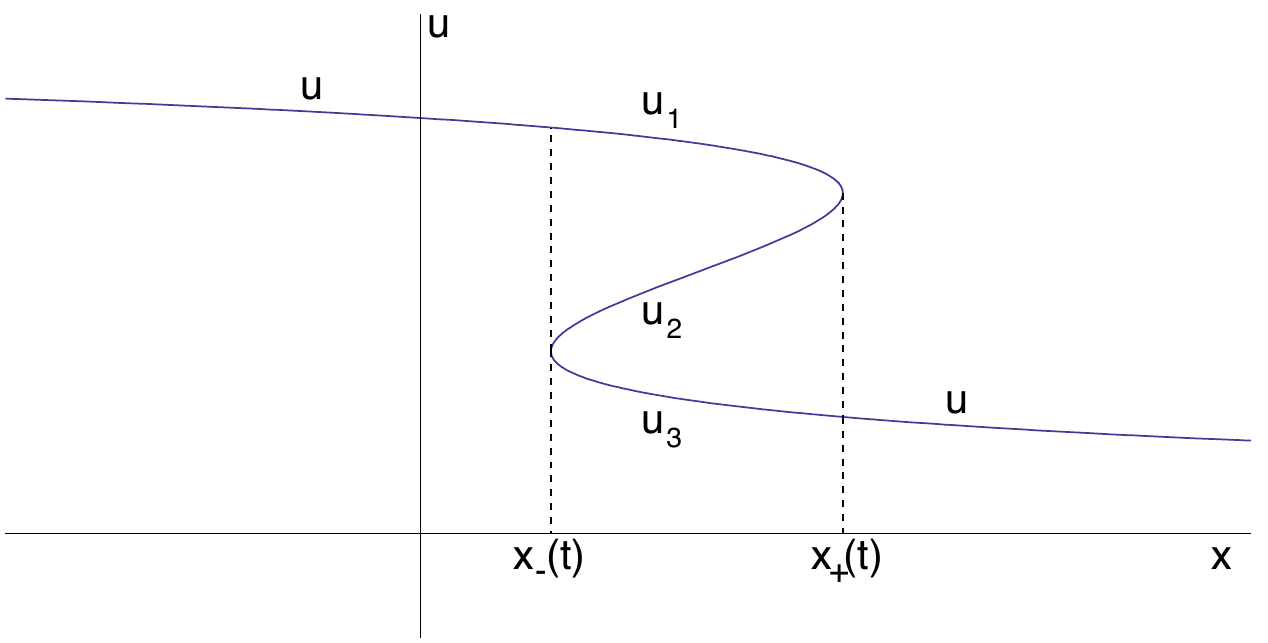}
\caption{Profile of the Burgers and Whitham solutions. The Burgers solution $u$ of (\ref{we0}) lives
in the single value regions while the Whitham solution $u_1$, $u_2$ and $u_3$ of (\ref{CHW}) reside in the multiple value region.}
\end{center}\end{figure}

The Burgers solution $u$ of (\ref{we0}) and the Whitham solution $u_1$, $u_2$, $u_3$
of (\ref{CHW})
must match on the trailing edge $x=x_{-}(t)$ and leading edge $x=x_+(t)$.
We see from Figure 1 that
\begin{eqnarray}
\label{bc1}
\left\{\begin{array}{c}
u_1= u  \\
u_2=u_3
\end{array}
\right.
\end{eqnarray}
must be imposed on the trailing edge, and that
\begin{eqnarray}
\label{bc2}
\left\{\begin{array}{c}
u_1=u_2   \\
u_3= u
\end{array}
\right.
\end{eqnarray}
must be satisfied on the leading edge.

In this paper, we consider the initial function $u_0(x)$ that is monotone. Since
the Burgers solution will not develop any shock if $u_0(x)$ is an increasing function,
we will focus on decreasing initial functions. Denoting the inverse function of
$u_0(x)$ by $f(u)$, the Burger equation (\ref{we0}) can be solved using the method
of characteristics; its solution is given implicitly by a hodograph transform,
\begin{equation}
\label{BS}
x = (3u + 2 \nu) t + f(u) \ .
\end{equation}

The solution method (\ref{BS}) has been generalized to solve the first order quasilinear
hyperbolic equations which can be written in Riemann invariant form and which are
strictly hyperbolic
\begin{equation}
\label{hydro}
{\partial u_i \over \partial t} + s_i(u_1, \cdots, u_n) {\partial u_i \over
\partial x} = 0 \ , \quad i=1,2,\cdots,n.
\end{equation}
The strict hyperbolicity means that the wave propagation speeds $s_i$'s do not
coincide.

We include Tsarev's theorem
for completeness \cite{Tian1, tsa}.
\begin{thm}
\label{th1}
If $w_{i}(u_1,u_2, \cdots, u_n)$'s solve the linear equations
\begin{equation}
\label{wi}
{\partial w_i \over \partial u_j}=A_{ij}(w_i-w_j)
\end{equation}
with
\begin{equation}
\label{Aij}
A_{ij}={ {\partial s_j \over \partial u_j} \over s_i-s_j}
\end{equation}
for $i,j=1,2,\cdots,n$ and $i \neq j$, then the solution
$(u_{1}(x,t), \cdots, u_{n}(x,t))$ of the hodograph
transform,
\begin{equation}
\label{ho}
x = s_{i}(u_1,u_2, \cdots, u_n) t + w_{i}(u_1,u_2, \cdots, u_n)
\end{equation}
satisfies equations (\ref{hydro}).
Conversely, any solution $(u_{1}(x,t), \cdots, u_{n}(x,t))$ of equations
(\ref{hydro}) can be obtained in this way in the neighborhood of
$(x_0,t_0)$ at which $u_{ix}$ are not all vanishing.
\end{thm}

The strict hyperbolicity, i.e., $s_i \neq s_j$ for  $i \neq j$, of (\ref{hydro}) is assumed to ensure that $A_{ij}$'s
of (\ref{Aij}) are not singular.

The result is classical when $n=2$.

The validity of this theorem
hinges on two factors. First, the linear equations (\ref{wi}) must have
solutions. Secondly, the hodograph transform (\ref{ho}) must not be degenerate,
i.e., it can be solved for $u_i$'s as functions of $x$ and $t$.
One interesting observation is that the Jacobian matrix of (\ref{ho})
is always diagonal on the solution $(u_{1}(x,t), \cdots, u_{n}(x,t))$.
\begin{cor}
\label{c}
At the solution $(u_{1}(x,t), \cdots, u_{n}(x,t))$ of
$x = s_{i}(u_1,u_2, \cdot, u_n) t + w_{i}(u_1,u_2, \cdots, u_n)$, $i=1,2, \cdots,
n$, the partial derivatives
\begin{displaymath}
\frac{\partial (s_{i} t + w_{i})}{\partial u_{j}} = 0
\end{displaymath}
for $i \neq j$.
\end{cor}
\begin{proof}
\begin{eqnarray*}
\frac{\partial (s_{i} t + w_{i})}{\partial u_{j}} & = & \frac{\partial
s_{i}}{\partial u_{j}}t + \frac{\partial w_{i}}{\partial u_{j}}
 \\
     & = & A_{ij} [ (s_{i} t +
w_{i}) - (s_{j} t + w_{j})]  \\
     & = & 0 \ .
\end{eqnarray*}
\end{proof}

Another aspect of Theorem \ref{th1} is that it is a local result. Solutions
produced by the hodograph transform are, in general, local in nature.
However, global solutions can
still be obtained if the conditions of the theorem are satisfied globally \cite{GT, Tian1}.

The Whitham equations (\ref{CHW}) will be shown to be non-strictly hyperbolic,
i.e., $\lambda_i$'s coincide at some points $(u_1, u_2, u_3)$ where $0<u_3 + \nu
< u_2 + \nu < u_1 + \nu$. However, Theorem \ref{th1} can still be applied to
equations (\ref{CHW}) since the functions $A_{ij}$'s of (\ref{Aij}) are still
non-singular for the Whitham equations (\ref{CHW}), even at the points of
non-strict hyperbolicity.
\begin{lem}
\label{Bij}
$$B_{ij} := {{\partial \lambda_i \over \partial u_j} \over \lambda_i-\lambda_j}
= {1 \over 2}{ (\lambda_i-\gamma) -2(u_i-u_j) \over
(\lambda_j-\gamma)(u_i-u_j)} \ ,  \quad i \neq j \,, $$
where $\gamma = u_1 + u_2 + u_3 + 2 \nu$.
\end{lem}

\begin{proof}

By (\ref{I}), we calculate
\begin{eqnarray*}
\lambda_i-\lambda_j&=&I{\partial_{u_i}I-\partial_{u_j}I\over
(\partial_{u_i}I) (\partial_{u_j}I)}=2I{(u_i-u_j)
\partial_{u_i u_j}^2 I \over
(\partial_{u_i}I) (\partial_{u_j} I)} \ ,\\
{\partial \lambda_i \over \partial u_j}&=&
{2(u_i-u_j)\partial_{u_i u_j}^2 I \over \partial_{u_i}I}+
{I \partial_{u_i u_j}^2 I \over (\partial_{u_i}I)^2} \ ,
\end{eqnarray*}
where we have used (\ref{EPD}). Hence, we get
\begin{eqnarray*}
\nonumber
B_{ij}
&=&{ 2(u_i-u_j)+ {I \over \partial_{u_i}I} \over
{2I \over \partial_{u_j}I}(u_i-u_j)}  \\
\label{symmetry}
&=&{1 \over 2}{ (\lambda_i-\gamma) - 2(u_i-u_j) \over
(\lambda_j-\gamma)(u_i-u_j)} \ ,
\end{eqnarray*}
where we have used (\ref{I}) to express $I /\partial_{u_i} I$.
\end{proof}

\section{The Single Phase Whitham Equations}

In this section, we will summarize some of the properties of the speeds
$\lambda_i$'s of (\ref{I}) for later use.

Function $I$ of (\ref{I}) is a complete elliptic integral; indeed,
\begin{equation}
\label{pi}
I(u_1, u_2, u_3) = {2(u_3 + \nu) \Pi(\rho, s) \over \sqrt{(u_1 - u_3)(u_2 + \nu)} } \ ,
\end{equation}
where $\Pi(\rho, s)$ is the complete integral of third kind, and
\begin{equation} \label{s}
\rho = {u_2 - u_3 \over u_2 + \nu} \ , \quad  s = \frac{(u_2 - u_3)(u_1 + \nu)}{(u_1 - u_3)(u_2 + \nu)} \ .
\end{equation}

Properties of complete elliptic integrals of the first, second and third kind are listed
in Appendix A.

Using the well known derivative formulae (\ref{P3}) and (\ref{P3'}),
one is able to rewrite $\lambda_i$ of (\ref{I}) as \cite{abe}
\begin{eqnarray}
\lambda_1(u_1, u_2, u_3) &=& u_1 + u_2 + u_3 + 2 \nu  + 2(u_1 - u_2)
\frac{(u_3 + \nu) \Pi(\rho,s)}{(u_2 + \nu)E(s)} \ , \nonumber \\
\lambda_2(u_1, u_2, u_3) &=& u_1 + u_2 + u_3 + 2 \nu  + 2 (u_3 - u_2)
\frac{ (1-s)\Pi(\rho, s) }{E(s) - (1-s) K(s)} \ , \label{lambda} \\
\lambda_3(u_1, u_2, u_3) &=& u_1 + u_2 + u_3 + 2 \nu + 2(u_2-u_3)
\frac{(u_3 + \nu) \Pi(\rho, s)}{(u_2 + \nu)[E(s) - K(s)]} \ . \nonumber
\end{eqnarray}
Here $K(s)$ and $E(s)$ are complete elliptic integrals of the first
and second kind.

Using inequalities (\ref{KE}), we obtain
\begin{eqnarray}
\lambda_{1} - (u_{1} + u_{2} + u_{3} + 2 \nu) &>& 0 \ , \label{l1>}  \\
\lambda_{2} - (u_{1} + u_{2} + u_{3} + 2 \nu) &<& 0 \ , \label{l2<}  \\
\lambda_{3} - (u_{1} + u_{2} + u_{3} + 2 \nu) &<& 0 \ , \label{l3<}
\end{eqnarray}
for $u_1 > u_2 > u_3> - \nu$.
In view of (\ref{K}-\ref{E2}) and (\ref{P1}-\ref{P2}), we find that $\lambda_{1}$,
$\lambda_{2}$ and
$\lambda_{3}$ have behavior

(1) At $u_{2}$ = $u_{3}$,
\begin{equation}
\label{tr}
\begin{array}{ll}
\lambda_{1}(u_1, u_2, u_3) = 3 u_{1} + 2 \nu\ , \\
\lambda_{2}(u_1, u_2, u_3) =
\lambda_{3}(u_1, u_2, u_3) = u_1 + 2 u_3 + 2 \nu - {4 (u_3 + \nu)(u_1 - u_3) \over u_1 + \nu} \ .
\end{array}
\end{equation}

(2) At $u_{1}$ = $u_{2}$,
\begin{equation}
\label{le}
\begin{array}{ll}
\lambda_{1}(u_1, u_2, u_3) =
\lambda_{2}(u_1, u_2, u_3) = 2 u_{1} +  u_{3}  + 2 \nu \ , \\
\lambda_{3}(u_1, u_2, u_3) = 3 u_{3} + 2 \nu \ .
\end{array}
\end{equation}

\begin{lem}
\label{IM}
\begin{equation*}
\frac{\pd \lambda_3 }{\pd u_3} < \frac{3}{2}
\frac{\lambda_2 - \lambda_3 }{u_2 - u_3} < \frac{\pd
  \lambda_2}{\pd u_2}
\end{equation*}
for $0< u_3 + \nu  < u_2 + \nu < u_1 + \nu$.
\end{lem}
\begin{proof}

Comparing formulae (\ref{I}) and (\ref{lambda}), we use (\ref{pi}) to obtain
\begin{equation}
\label{lambda2}
{\pd I \over \pd u_2} ={ \sqrt{(u_1 - u_3)(u_2 + \nu)}  \over (u_2 - u_3) (u_1 - u_2)} \ [ E - (1-s) K] \ .
\end{equation}
Differentiating (\ref{I}) for $\lambda_2$ and using (\ref{lambda2}) yields
\begin{eqnarray}
{\pd \lambda_2 \over \pd u_2} &=& {I \pd^2_{u_2 u_2} I \over (\pd_{u_2} I)^2}  \no  \\
&=& {(u_3 + \nu) \Pi \over (u_2 + \nu) [E - (1-s) K]^2} \times \no
\\ & & \times
\left \{ [s + {u_2 - u_3 \over u_1 - u_3} - 2
{u_1 - u_2 \over u_1 - u_3}]
[E - (1-s) K] + s(1-s) K \right \}  \label{lambda2u2} \ .
\end{eqnarray}
Using formulae (\ref{lambda}) for
$\lambda_2$ and $\lambda_3$, we obtain
\begin{equation}
\label{M}
\lambda_2(u_1, u_2, u_3) - \lambda_3(u_1, u_2, u_3) = {2 (u_2 - u_3) (u_3 + \nu) \Pi 
M(u_1, u_2, u_3)\over (u_2 + \nu)(K-E)[E-(1-s)K]} \ ,
\end{equation}
where
\begin{equation}
\label{M'}
M(u_1, u_2, u_3) = [1 + {u_1 - u_2 \over u_1 - u_3}] E - [{u_1 - u_2 \over u_1 - u_3} + (1-s)] K \ .
\end{equation}

We then obtain from (\ref{lambda2u2}) and (\ref{M}) that
\begin{align*}
{\pd \lambda_2 \over \pd u_2} - {3 \over 2} \ {\lambda_2 - \lambda_3 \over u_2 - u_3}
&= - {(u_3 + \nu)(4-3s) \Pi E^2 \over (u_2 + \nu) [K - E][E - (1-s) K]^2} \left\{ (1-s) ({K \over E})^2  
\right. \\ 
&\phantom{=}
\left. 
\quad -2 ({K \over E}) + {4 + s \over 4 -3s} \right\}  > 0 \ ,
\end{align*}
where the inequality follows from the negativity of the function in the bracket (c.f. (4.18) of \cite{Tian1}).
This proves part of the lemma. The other part can be shown in the same way.

\end{proof}

We conclude this section with a few calculations.
We use (\ref{K3}) and (\ref{E3}) to calculate the derivative of (\ref{M'})
\begin{eqnarray}
\label{pM}
{\pd M(u_1,u_2,u_3) \over \pd u_2} &=& {1 \over 2(u_1 - u_3)} \ \{ [2 + {(u_1 + \nu)(u_3 + \nu) \over (u_2 + \nu)^2}] K  \no
\\ & & \phantom{ {1 \over 2(u_1 - u_3)} \ \{ } 
-[2 + {u_1 + \nu \over u_2 + \nu}] E \}  \ , \\
{\pd^2 M(u_1,u_2,u_3) \over \pd u_2^2} &=& {(u_1 + \nu) [ 4(u_1 + \nu) - 2 (u_2 + \nu) + (u_3 + \nu)] E
 \over 4(u_1 - u_3)(u_1 - u_2) (u_2 + \nu)^2}  \ . \label{ppM}
 \end{eqnarray}

Finally, we use the expansions (\ref{K}-\ref{E}) for $K$ and $E$ to obtain
\begin{eqnarray}
M(u_1, u_2,u_3) &=& {\pi \over 2} \left \{ {(u_2 - u_3)s  \over 2(u_1 - u_3)} + {1 \over 16} \ (1 - {3 (u_1 - u_2) \over u_1 - u_3}) s^2  \no 
\right. \\ & & \phantom{{\pi \over 2} } \left. \quad +
{3 \over 128} (1 - {5(u_1 - u_2)
\over u_1 - u_3} ) s^3 + O(s^4) \right \}
\nonumber \\
&=& {\pi \over 2} \left \{ [ {u_2 + \nu \over 2(u_1 + \nu)} + {1 \over 16} \ (1 - {3 (u_1 - u_2)
\over u_1 - u_3}) ] s^2  \no 
\right. \\ & & \phantom{{\pi \over 2}} \left. \quad +
{3 \over 128} (1 - {5(u_1 - u_2)
\over u_1 - u_3} ) s^3 + O(s^4) \right \} \,,
\label{M2}
\end{eqnarray}
where we have used formula (\ref{s}) for $s$ in the last equality.

\section{Step-like Initial Data}

In this section, we will consider
the step-like initial data
\begin{equation} \label{step}
u_0(x) = \left\{ \begin{matrix} a & \quad  x < 0 \\
b & \quad x > 0 \end{matrix} \right. \ , \ \quad a \neq b
\end{equation}
for
equation (\ref{CH}).
Since the solution of (\ref{we0}) will never develop a shock when $a \leq b$, we will
be interested only in the case $a > b$.  We classify the initial data (\ref{step}) into two types:
\begin{itemize}
\item{(I)} $a + \nu > 4(b + \nu)$  \ ,
\item{(II)} $a + \nu \leq 4(b + \nu) \,.$
\end{itemize}
We will solve the initial value problem for the Whitham equations for these two types of initial data.

\subsection{Type I: $
a + \nu > 4(b + \nu)$}

\begin{figure}[h]
\begin{center}
\includegraphics[width=12cm]{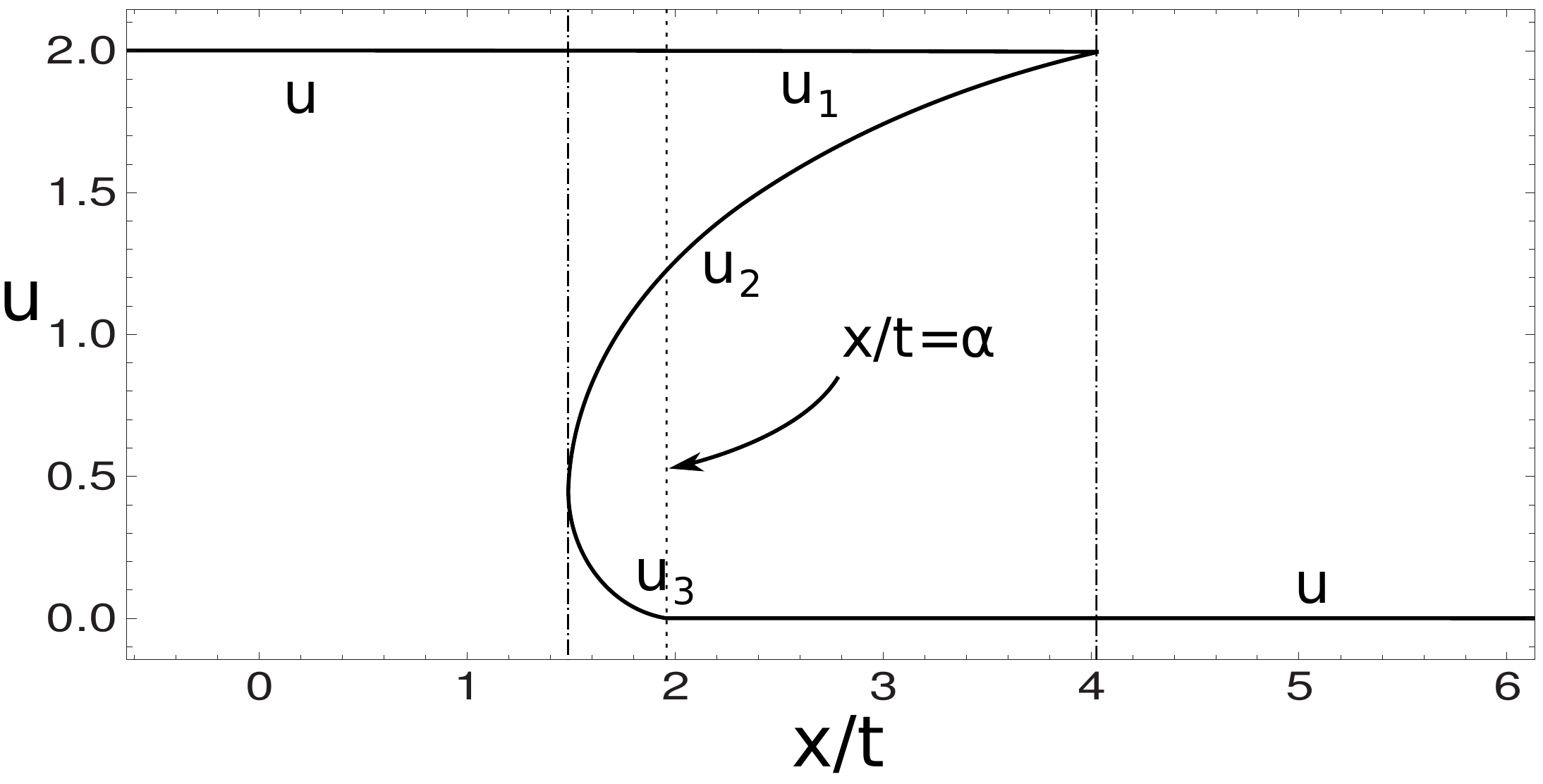}
\caption{Self-Similar solution of the Whitham equations for
  $a=2$, $b= 0$ and $\nu = 1/20$ of type I.}
\end{center}\end{figure}

\begin{thm}(see Figure 2.)
\label{4.1}For the step-like initial data $u_0(x)$ of (\ref{step}) with $0 < b + \nu < (a + \nu)/4$, the solution of the Whitham equations
(\ref{CHW}) is given by
\begin{equation}
\label{ws1}
u_1 = a \ , \quad x = \lambda_2(a, u_2, u_3) \ t \ , \quad x = \lambda_3(a, u_2, u_3) \ t
\end{equation}
for $(3a - \nu)/4 < x/t \leq \alpha$ and by
\begin{equation}
\label{ws2}
u_1 = a \ , \quad x = \lambda_2(a, u_2, b) \ t \ , \quad u_3 = b
\end{equation}
for $\alpha t \leq x < (2a + b + 2 \nu) t$, where $\alpha = \lambda_2(a, u^*, b)$ and $u^*$ is the unique
solution $u_2$ of $\lambda_2(a,u_2,b)=\lambda_3(a,u_2,b)$ in the interval $b<u_2<a$.
Outside the region $(3a - \nu)/4 < x/t < 2a + b + 2 \nu$, the solution of the Burgers
equation (\ref{we0}) is given by
\begin{equation}
\label{bs1}  u \equiv a \quad \mbox{$x/t \leq (3a - \nu)/4$}
\end{equation}
and
\begin{equation}
\label{bs2}  u \equiv b \quad \mbox{$x/t \geq 2a + b + 2 \nu$} \ .
\end{equation}
\end{thm}

The boundaries $x/t = (3a -  \nu)/4$ and $x/t = 2a + b + 2 \nu$ are
the trailing and
leading edges, respectively, of the dispersive shock.
They separate the solution into the region governed by the single phase
Whitham equations  and the region governed by the Burgers equation.

The proof of Theorem \ref{4.1} is based on a series of lemmas.

We first show that the solutions defined by formulae (\ref{ws1}) and (\ref{ws2})
indeed satisfy the Whitham equations (\ref{CHW}) \cite{tsa}.

\begin{lem}
\label{4.2}
\begin{enumerate}
\item[(i)] The functions $u_1$, $u_2$ and $u_3$ determined by equations (\ref{ws1})
give a solution of the Whitham equations (\ref{CHW}) as long as $u_2$ and $u_3$
can be solved from (\ref{ws1}) as functions of $x$ and $t$.

\item[(ii)] The functions $u_1$, $u_2$ and $u_3$ determined by equations (\ref{ws2})
give a solution of the Whitham equations (\ref{CHW}) as long as $u_2$
can be solved from (\ref{ws2}) as a function of $x$ and $t$.
\end{enumerate}

\end{lem}

\begin{proof}

(i) $u_1$ obviously satisfies the first equation of (\ref{CHW}). To verify the
second and third equations, we observe that
\begin{equation}
\label{dia}
\frac{\pd \lambda_2 }{\pd u_3} = \frac{\pd \lambda_3 }{\pd u_2} = 0
\end{equation}
on the solution of (\ref{ws1}). This follows from Lemma \ref{Bij}.

We then calculate the partial derivatives of the second equation of (\ref{ws1})
with respect to $x$ and $t$.
$$ 1 = \frac{\pd \lambda_2 }{\pd u_2} \ t ( u_{2} )_x \ ,
\quad 0 = \frac{\pd \lambda_2 }{\pd u_2} \ t (u_{2})_t + \lambda_2 \ ,$$
which give the second equation of (\ref{CHW}).

The third equation of (\ref{CHW}) can be verified in the same way.

(ii) The second part of Lemma \ref{4.2} can easily be proved in a similar
manner.

\end{proof}

We now determine the trailing edge. Eliminating $x$ and $t$ from the last two equations of (\ref{ws1})
yields
\begin{equation}
\label{m23}
\lambda_2(a, u_2, u_3) - \lambda_3(a, u_2, u_3) = 0 \ .
\end{equation}
In view of formula (\ref{M}),  we replace (\ref{m23}) by
\begin{equation}
\label{F1}
{M(a, u_2,u_3) \over s^2} = 0 \ .
\end{equation}

Therefore, at the trailing edge where $u_2=u_3$, i.e., $s=0$, equation
(\ref{F1}), in view of the expansion (\ref{M2}), becomes
$${u_3 + \nu \over 2(a + \nu)} + {1 \over 16} \ [1 - {3 (a - u_3) \over a - u_3}] = 0 \ ,$$
which gives $u_2=u_3 = (a - 3 \nu)/4$.

\begin{lem} \label{4.3}
Equation (\ref{F1}) has a unique solution satisfying $u_2=u_3$. The solution
is $u_2=u_3=(a - 3 \nu)/4$. The rest of equations (\ref{ws1}) at the trailing edge
are $u_1=a$ and
$x/t = \lambda_2(a, (a - 3 \nu)/4, (a - 3 \nu)/4) = (3a - \nu)/4$.
\end{lem}

Having located the trailing edge, we now solve equations (\ref{ws1})
in the neighborhood of the trailing edge. We first consider equation
(\ref{F1}). We use (\ref{M2}) to differentiate $M/s^2$ at the
trailing edge $u_1=a$, $u_2=u_3=(a-3 \nu)/4$, to find
$${\pd \over \pd u_2} [{M \over s^2}] =
{\pd \over  \pd u_3} [{M \over s^2}] = {\pi \over 8(a+\nu)}  \ ,$$
which shows that equation (\ref{F1}) or equivalently (\ref{m23}) can be
inverted to give $u_3$ as a decreasing
function of $u_2$
\begin{equation}
\label{A} u_3 = A(u_2)
\end{equation}
in a neighborhood of $u_2=u_3=(a - 3\nu)/4$.

We now extend the solution (\ref{A}) of equation (\ref{m23}) in the region
$a > u_2 > (a - 3\nu)/4 > u_3 > b$ as far as possible. We deduce from Lemma \ref{IM} that
\begin{equation}
\label{dia2}
{\pd \lambda_2 \over \pd u_2} > 0 \ , \quad {\pd \lambda_3 \over \pd u_3} <  0
\end{equation}
on the solution of (\ref{m23}).
Because of (\ref{dia}) and (\ref{dia2}), solution (\ref{A}) of equation (\ref{m23})
can be extended as long as $a > u_2 > (a - 3\nu)/4 > u_3 > 0$.

There are two possibilities: (i) $u_2$ touches $a$ before or simultaneously
as $u_3$ reaches $b$ and (ii) $u_3$ touches $b$ before $u_2$ reaches $a$.

It follows from (\ref{le}) that
$$\lambda_2(a,a,u_3) > \lambda_3(a,a,u_3) \quad \mbox{for $b \leq u_3 < a$} \ .$$
This shows that (i) is impossible. Hence, $u_3$ will touch $b$ before $u_2$
reaches $a$. When this happens, equation (\ref{m23}) becomes
\begin{equation}
\lambda_2(a, u_2, b) - \lambda_3(a, u_2, b) = 0 \ . \label{u2}
\end{equation}

\begin{lem} \label{4.4}
Equation (\ref{u2}) has a simple zero in the region $b < u_2 < a$, counting
multiplicities. Denoting the zero by $u^*$, then $\lambda_2(a, u_2, b) - \lambda_3(a, u_2, b)$ is positive
for $u_2 > u^*$ and negative for $u_2 < u^*$.
\end{lem}

\begin{proof}
We use (\ref{M}) and (\ref{ppM}) to prove the lemma.
In equation (\ref{M}), $K-E$ and $E-(1-s)K$ are all positive for $0<s<1$ in view of (\ref{KE}).
We claim that
$$M(a, u_2, b) = 0 \quad \mbox{for $u_2 = b$} \ , \quad M(a, u_2, b) < 0 \quad
\mbox{for $u_2$ near $b$} \ , $$
and $$
\quad M(a, u_2, b) > 0 \quad \mbox{for $u_2=a$} \ .$$
The equality and the first inequality follow from expansion (\ref{M2}) and
$a + \nu > 4(b+\nu)$. The second inequality is
obtained by applying (\ref{K2}) and (\ref{E2}) to (\ref{M'}).

We conclude from the two inequalities that $M(a, u_2, b)$ has a zero
in $b < u_2 < a$. This zero is unique because $M(a, u_2, b)$, in
view of (\ref{ppM}), is a convex function of $u_2$. This zero is
exactly $u^*$ and the rest of the theorem is proven easily.
\end{proof}

Having solved equation (\ref{m23}) for $u_3$ as a decreasing
function of $u_2$ for $(a - 3 \nu)/4 \leq u_2 \leq u^*$, we turn to
equations (\ref{ws1}). Because of (\ref{dia}) and (\ref{dia2}), the
second equation of (\ref{ws1}) gives $u_2$ as an increasing function
of $x/t$, for $(3a -  \nu)/4 \leq x/t \leq \alpha$, where
\begin{equation*}
\alpha = \lambda_2(a, u^*, b).
\end{equation*}
Consequently, $u_3$ is a decreasing function of $x/t$ in the same interval.

\begin{lem} \label{4.5}
The last two equations of (\ref{ws1}) can be inverted to give $u_2$ and $u_3$ as
increasing and decreasing functions, respectively, of the self-similarity variable
$x/t$ in the interval $(3a - \nu) \leq x/t \leq \alpha$, where $\alpha = \lambda_2(a, u^*, b)$
and $u^*$ is given in Lemma \ref{4.4}.
\end{lem}

We now turn to equations (\ref{ws2}). We want to solve the second equation
when $x/t > \alpha$ or equivalently when $u_2 > u^*$. According to Lemma \ref{4.4},
$\lambda_2(a, u_2, b) - \lambda_3(a, u_2,b) > 0$ for $u^* < u_2 < a$, which, by 
Lemma \ref{IM},
shows that
$${\pd \lambda_2(a, u_2, b) \over \pd u_2} > 0 \ .$$
Hence, the second equation of (\ref{ws2}) can be solved for $u_2$ as an increasing
function of $x/t$ as long as $u^* < u_2 < a$. When $u_2$ reaches $a$, we have
$$x/t = \lambda_2(a, a, b) = 2a + b + 2 \nu \ ,$$
where we have used (\ref{le}) in the last equality. We have
therefore proved the following result.

\begin{lem} \label{4.6}
The second equation of (\ref{ws2}) can be inverted to give $u_2$ as an increasing
function of $x/t$ in the interval $\alpha \leq x/t \leq 2a + b + 2 \nu$.
\end{lem}

We are ready to conclude the proof of Theorem \ref{4.1}.

The Burgers solutions (\ref{bs1}) and (\ref{bs2}) are trivial.

According to Lemma \ref{4.5}, the last two equations of (\ref{ws1}) determine $u_2$
and $u_3$ as functions of
$x/t$ in the region $(3a - v)/4 \leq x/t \leq \alpha$. By the first part of Lemma \ref{4.2}, the
resulting $u_1$, $u_2$ and $u_3$ satisfy the Whitham equations (\ref{CHW}).
Furthermore, the boundary condition (\ref{bc1}) is satisfied
at the trailing edge $x/t = (3a - v)/4$.

Similarly, by Lemma \ref{4.6}, the second equation of (\ref{ws2}) determines $u_2$
as a function of $x/t$ in the region $\alpha \leq x/t \leq 2a + b + 2 \nu$. It then follows from
the second part of Lemma \ref{4.2} that $u_1$, $u_2$ and $u_3$ of (\ref{ws2}) satisfy
the Whitham equations (\ref{CHW}).
They also satisfy the boundary condition (\ref{bc2}) at the
leading edge $x = (2a + b + 2 \nu) t$.

We have therefore completed the proof of Theorem \ref{4.1}.

A graph of the Whitham solution for the initial data (\ref{step}) of type I is
given in Figure 2. It is obtained by
plotting the exact solutions of (\ref{ws1}) and (\ref{ws2}).

\subsection{Type II: $
a + \nu \leq 4(b + \nu)$}

\begin{figure}[h]
\begin{center}
\includegraphics[width=12cm]{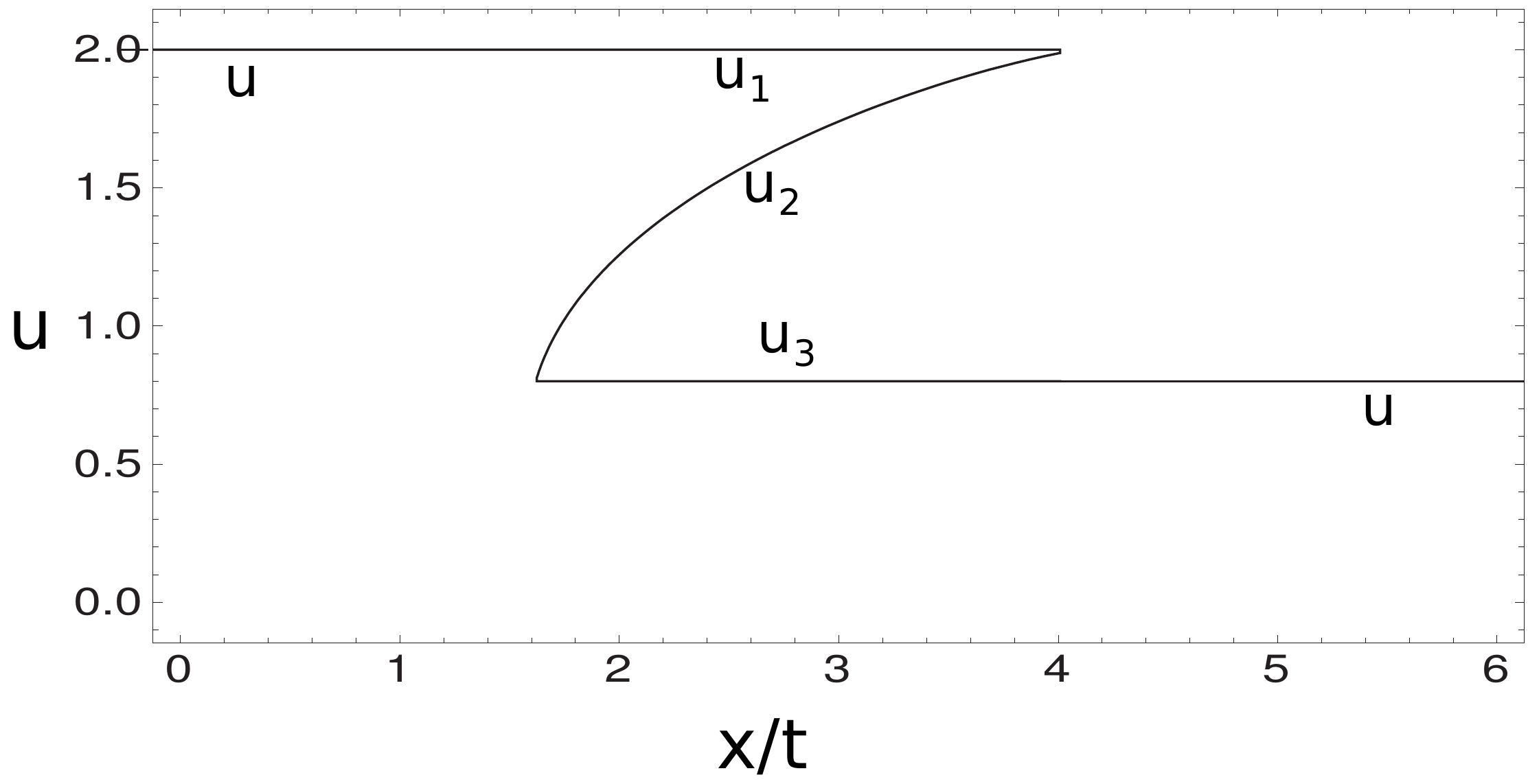}
\caption{Self-Similar solution of the Whitham equations for
  $a=2$, $b= 0.8$ and $\nu = 1/20$ of type II.}
\end{center}\end{figure}

\begin{thm}(see Figure 3.) \label{4.7}
For the step-like initial data (\ref{step})
with $ 0 < (a+\nu)/4 \leq b + \nu < a+ \nu$,
the solution of the
Whitham equations (\ref{CHW}) is given by
\begin{equation} \label{4.13}
u_1 =a \ , \quad x = \lambda_2(a, u_2, b) \ t \ , \quad u_3 = b
\end{equation}
for $\lambda_2(a,b,b) < x/t < \lambda_2(a,a,b)$, where $\lambda_2(a,b,b) = a + 2b +
2 \nu - 4(a - b)(b + \nu)/(a + \nu) $ and
$\lambda_2(a,a,b) = 2 a + b + 2 \nu$. Outside this interval, the solution of
(\ref{we0}) is given by
\begin{equation*}
u \equiv a \quad \mbox{$x/t \leq \lambda_2(a,b,b)$}
\end{equation*}
and
\begin{equation*}
u \equiv b \quad \mbox{$x/t \geq \lambda_2(a,a,b)$} \ .
\end{equation*}
\end{thm}

\begin{proof}

We will give a  brief proof, since the arguments are, more or less, similar to those in the proof of
Theorem \ref{4.1}.

It suffices to show that $\lambda_2(a, u_2, b)$ is an increasing function of $u_2$ for
$b < u_2 < a$. Using the inequality (\ref{KE}) to estimate the right hand side of
(\ref{pM}),
we obtain
\begin{align*}
{d M(a, u_2, b) \over d u_2} & >  {(a + \nu)(u_2 - b) E(s) \over
2 (2 - s)(a - b)^2 (u_2 + \nu)^2}  \left \{ 2(u_2 + \nu) + 2(b + \nu) - \right.
\\
&\phantom{>  {(a + \nu)(u_2 - b) E(s) \over
2 (2 - s)(a - b)^2 (u_2 + \nu)^2} } \left. -
(a + \nu)
\right \} > 0 
\end{align*}
for $b<u_2<a$, where we have used $(a + \nu)/4 \leq b + \nu$ in the second inequality.
Since $M(a,u_2,b) = 0$
at $u_2 = b$ in view of (\ref{M2}), this implies that
$M(a,u_2,b) > 0$
for $b < u_2 < a$. It then follows from (\ref{M}) that $\lambda_2(a, u_2, b) - \lambda_3(a, u_2, b) > 0$.
By Lemma \ref{IM}, we conclude that
$${d \lambda_2(a, u_2, b) \over d u_2} > 0$$ for $b < u_2 < a$.

\end{proof}

A graph of the Whitham solution for initial data (\ref{step}) of type II is given
in Figure 3.  It is obtained by plotting the exact solution of (\ref{4.13}).

\section{Smooth Initial Data}

In this section, we will study the initial value problem of the Whitham equations when
the initial values are given by a smooth monotone function $u_0(x)$. Since the Burgers solution
of (\ref{we0}) will never develop a shock when $u_0(x)$ is an increasing function, we will be
interested only in the case that $u_0(x)$ is a decreasing function.

We consider the initial function $u_0(x)$ which is a decreasing function and is bounded at
$x = \pm \infty$
\begin{equation}
\label{infty}
\lim_{x \rightarrow - \infty} u_0(x) = a \ , \quad \lim_{x \rightarrow + \infty} u_0(x) = b \ .
\end{equation}

By Theorem \ref{th1} and Lemma \ref{Bij}, we can use the hodograph transform,
\begin{equation}
\label{HO}
x = \lambda_i(u_1, u_2, u_3) t + w_i(u_1, u_2, u_3) \ , \quad i=1, 2, 3,
\end{equation}
to solve the Whitham equations (\ref{CHW}). Here, $w_i$'s satisfy a
linear over-determined system of type (\ref{wi})
\begin{eqnarray}
\label{wib}
{\partial w_i \over \partial u_j}=B_{ij}(w_i-w_j) \ ,
\end{eqnarray}
where $B_{ij}$'s are given in Lemma \ref{Bij}.

The boundary conditions on $w_i$'s are obtained by observing that the hodograph
solution (\ref{HO}) of the Whitham equations (\ref{CHW}) must match the characteristic
solution (\ref{BS}) of the Burgers equation (\ref{we0})
at the trailing and leading edges in the fashion of (\ref{bc1}-\ref{bc2}).
By (\ref{tr}-\ref{le}), $w_i$'s must satisfy the boundary conditions,
\begin{eqnarray}
\label{wbc1}
\left\{\begin{array}{l}
w_1(u_1, u_1, u_3)=
w_2(u_1, u_1, u_3)  , \\
w_3(u_1, u_1, u_3) = f(u_3)
 ,
\end{array}
\right. \\
\label{wbc2}
\left\{\begin{array}{l}
w_1(u_1, u_3, u_3)=
f(u_1) , \\
w_2(u_1, u_3, u_3)=
w_3(u_1, u_3, u_3)  ,
\end{array}
\right.
\end{eqnarray}
where $f(u)$ is the inverse of the initial function $u_0(x)$.

Analogous to the KdV case \cite{Tian1, Tian3}, equations (\ref{wib})
subject to
boundary conditions (\ref{wbc1}-\ref{wbc2}) are related to a boundary value
problem
of a linear over-determined system of Euler-Poisson-Darboux type (cf. (\ref{EPD}))
\begin{eqnarray}
\label{q}
2(u_i-u_j)\frac{\partial^2 q}{\partial{u_i}\partial{u_j}} & = &
\frac{\partial q}{\partial u_i}-\frac{\partial q}{\partial u_j} \ , \quad i, j = 1, 2, 3, \\
\label{bc}
q(u, u, u)&=& f(u) \ ,
\end{eqnarray}
for $i,j=1,2,3$. The solution is unique and symmetric in $u_1$, $u_2, u_3$.
It is given explicitly by \cite{Tian1}
\begin{eqnarray}
\label{qbeta}
q(u_1,u_2,u_3) = \frac{1}{2\sqrt{2}\pi}\int_{-1}^{1}\int_{-1}^{1}
\frac{f({1+u \over 2}{1+v \over 2}u_1 + {1+u \over 2}{1-v \over 2}u_2
+ {1-u \over 2}u_3)}{\sqrt{(1-u)(1-v^2)}}dudv \ .
\end{eqnarray}
\begin{thm}
\label{5.1}
If $q(u_1, u_2, u_3)$ is a solution of (\ref{q}) and (\ref{bc}),
then
\linebreak
$w_i(u_1, u_2, u_3)$'s given by
\begin{eqnarray}
\label{w}
w_i\ &= [\lambda_i(u_1, u_2, u_3)- \gamma] \
{\partial q(u_1, u_2, u_3) \over
\partial u_i}
 +
 q(u_1, u_2, u_3) \ ,
\end{eqnarray}
where $\gamma = u_1 + u_2 +u_3 + 2 \nu$, solve equations (\ref{wib}) and satisfy boundary conditions
(\ref{wbc1}-\ref{wbc2}).
\end{thm}
\begin{proof}

By Lemma \ref{Bij} and (\ref{q}), we obtain
\begin{eqnarray}
\label{Old}
[\lambda_j -\gamma]\left[ { \partial q \over
\partial u_i} -{ \partial q \over
\partial u_j} \right] B_{ij}
=\left({ \partial q \over
\partial u_j} -{ \partial q \over
\partial u_i}\right)+
[\lambda_i -\gamma]{ \partial^2q \over \partial u_i \partial u_j}\  .
\end{eqnarray}
Using (\ref{w}), we calculate
\begin{eqnarray*}
{\partial w_{i} \over \partial u_j}&=&{\partial \lambda_i \over
\partial u_j}{\partial q \over \partial u_i}
+ \left[\lambda_i-\gamma\right]{\partial^2 q \over
\partial u_i\partial u_j}+ {\partial q \over \partial u_j} -{\partial q
\over \partial u_i}\ , \\
w_i-w_j&=&\left[\lambda_i-\gamma \right][{\partial q \over \partial
u_i}- {\partial q \over \partial u_j}] +
\left[\lambda_i- \lambda_j \right ]{\partial q \over \partial u_i} \ .
\end{eqnarray*}
Substituting these into (\ref{Old}), we find that $w_i$'s satisfy (\ref{wib}).

Finally, we shall check the boundary conditions (\ref{wbc1}-\ref{wbc2}).
We only consider the leading edge, and the trailing edge can be handled
in the same way.

Since $q(u_1, u_2, u_3)$ is symmetric in $u_1$,
$u_2$ and $u_3$, the first condition of (\ref{wbc1}) follows from
(\ref{w}) and (\ref{le}).

For the second condition, it follows from (\ref{le}) and (\ref{w}) again that
\begin{equation}
\label{w3}
w_{3}(u_{1},u_{1},u_{3}) = 2(u_{3}-u_{1})
{\partial q \over \partial u_{3}} + q \ .
\end{equation}
Differentiating this with respect to $u_{1}$ yields,
\begin{eqnarray*}
{\partial w_{3} \over \partial u_{1}} =  -2 {\partial q \over \partial
u_{3}} +2(u_{3} - u_{1}) [{\partial^2 q \over \partial u_1
\partial u_3} + {\partial^2 q \over \partial u_2 \partial u_3}]
+{\partial q \over \partial u_{1}}+{\partial q \over \partial
u_{2}} = 0 \ ,
\end{eqnarray*}
where we have used (\ref{q}) in the last equality. Since
$w_{3}(u_{1},u_{1},
u_{3})$ is independent of $u_{1}$, we replace $u_{1}$
by $u_{3}$ in (\ref{w3}) and use (\ref{bc}) to obtain the second
condition of (\ref{wbc1}).

\end{proof}

Theorem \ref{5.1} has been reported in \cite{abe} in the case of $\nu = 0$.

In the rest of this section, we study the hodograph transform
(\ref{HO}) with $w_{i}$'s given by
(\ref{qbeta}) and (\ref{w}). We shall show that the transform can be solved for
$u_{1}$,
$u_{2}$ and $u_{3}$ as functions of $(x, t)$ within a cusp in the
$x$-$t$ plane.

Since $t_{b} = - [3 \min\left( u'_{0}(x)\right)]^{-1}$ is the
breaking time of the Burgers solution of (\ref{we0}), the breaking
is caused by an inflection point in the initial data. If $x_{0}$ is
this inflection point, then $(x_{b}, t_{b})$ is the breaking point
on the evolving curve where $x_{b} = x_{0} + [3 u_0(x_{0}) + 2 \nu] t_b$,
and $t_{b}$ is the breaking time. Without loss of generality, we may
assume $x_b=0$, $t_b = 0$ and denote $u_0(0)$ by $\hat{u}$.  The
effect of these choices is that we are starting at the breaking
time, and the evolving curve is about to turn over at the point
$(0, \hat{u})$ in the $x$-$u$ plane. It immediately follows that
\begin{equation}
\label{cf}
f(\hat{u}) = f'(\hat{u}) = f''(\hat{u}) = 0 \ ,
\end{equation}
where $x = f(u)$ is the inverse function of the
decreasing initial data $u = u_{0}(x)$. On the assumption that
$x = f(u)$ has only one inflection point, it follows from the monotonicity
of the function $f(u)$ that
\begin{equation}
\label{cf2}
f''(u) = \left\{ \begin{array}{ll}
                   < 0   & \quad u > \hat{u} \\
                   = 0   & \quad u = \hat{u} \\
                   > 0   & \quad u < \hat{u}
                 \end{array}
         \right. .
\end{equation}
Under a little bit stronger condition than (\ref{cf2}), we will be able to show that hodograph
transform (\ref{HO}) can be inverted to give $u_{1}$, $u_{2}$ and $u_{3}$ as functions
of $(x,t)$ in some domain of the $x$-$t$ plane.

\begin{thm}
\label{main}
Suppose $u_0(x)$ is a decreasing function satisfying (\ref{infty}) with $a + \nu > b + \nu > 0$.
If, in addition to (\ref{cf}), the inverse function $f(u)$ satisfies $f'''(u)$ $<$ 0 for $b < u < a$, then transform
(\ref{HO}) with $w_{i}$'s given by (\ref{qbeta}) and (\ref{w}) can be solved for
$u_{1}$, $u_{2}$ and
$u_{3}$ as functions of $(x,t)$ within a cusp in the $x$-$t$ plane for
all $t > 0$. Furthermore, these $u_{1}$,
$u_{2}$ and $u_{3}$ satisfy boundary conditions (\ref{wbc1}-\ref{wbc2})
on the boundary of the cusp.
\end{thm}

The proof is based on a series of lemmas. The organization is as
follows: we eliminate $x$ from transform (\ref{HO}) to obtain two
equations involving $u_{1}$, $u_{2}$, $u_{3}$, and $t$. These two
equations can be shown, for each fixed time after the breaking, to
determine $u_{1}$ and $u_{3}$ as decreasing functions of $u_{2}$
within an interval whose end points depend on $t$. Substituting
$u_{1}$ and $u_{3}$ as functions of $u_{2}$ into the hodograph
transform, we find that, within a cusp in the $x$-$t$ plane, $u_{2}$
is a function of $(x, t)$, and so, therefore, are $u_{1}$ and
$u_{3}$.

First, we conclude from formula (\ref{qbeta})
\begin{lem}
\label{q3} If $f(u)$ satisfies the conditions of Theorem \ref{main},
then $q(u_1, u_{2}, u_{3})$ given by (\ref{qbeta}) satisfies
\begin{displaymath}
\frac{\partial^{3}q}{\partial u_{i} \partial u_{j} \partial u_{k}}
< 0 \ ,
\hspace*{.2in} i, j, k = 1, 2, 3 \ .
\end{displaymath}
\end{lem}

Eliminating $x$ from (\ref{HO}) yields
\begin{eqnarray}
(\lambda_{1}t + w_{1}) - (\lambda_{2}t + w_{2}) &=& 0 \ , \label{1} \\
 (\lambda_{2}t + w_{2}) - (\lambda_{3}t + w_{3}) &=& 0 \ . \label{2}
\end{eqnarray}

Using (\ref{w}) for $w_1$ and $w_2$, and (\ref{lambda}) for $\lambda_1$ and $\lambda_2$ , we write
\begin{eqnarray*}
(\lambda_{1}t + w_{1}) - (\lambda_{2}t + w_{2})  &=& (\lambda_1 - \gamma)[t + {\pd q \over \pd u_1}]
- (\lambda_2 - \gamma)[t + {\pd q \over \pd u_2}] \\
&=& {2(u_1 - u_2)(u_3 + \nu) \Pi \over (u_2 + \nu) E} \ F(u_1, u_2, u_3) \ ,
\end{eqnarray*}
where
\begin{equation}
\label{F}
F = (t+ {\pd q \over \pd u_1}) + {s E \over E - (1-s) K} \ {u_2 + \nu \over u_1 + \nu} \
[t + {\pd q \over \pd u_2}] \ .
\end{equation}

Similarly, we use (\ref{w}) for $w_2$ and $w_3$ to write
\begin{eqnarray*}
(\lambda_{2}t + w_{2}) - (\lambda_{3}t + w_{3}) &=& (\lambda_2 - \lambda_3)(t + {\partial q
\over \partial u_2}) + (\lambda_3 - \gamma)({\pd q \over \pd u_2} - {\pd q \over \pd u_3} ) \\
&=& {2(u_2 - u_3)(u_3 + \nu) \Pi \over (u_2 + \nu)[K - E][E - (1-s) K] } \ G(u_1, u_2, u_3) \ ,
\end{eqnarray*}
where
\begin{equation}
\label{G}
G =  M(u_1, u_2, u_3)
(t + {\pd q \over \pd u_2}) - 2(u_2 - u_3) [E - (1-s)K]{\pd^2 q \over \pd u_2 \pd u_3} \ .
\end{equation}
In the derivation, we have used formula (\ref{M}) for $\lambda_2 - \lambda_3$, formula (\ref{lambda}) for
$\lambda_3$ and equation (\ref{q}).

Since (\ref{KE}) implies that $K(s) - E(s) > 0$ and $E(s) - (1-s)
K(s)
> 0$, equations (\ref{1}) and (\ref{2}) are
equivalent to
\begin{equation}
\label{FG} F(u_1, u_2, u_3) = 0 \ , \quad G(u_1, u_2, u_3) = 0
\end{equation}
for $0 < s < 1$.

\subsection{The trailing edge}

\smallskip

We first study the trailing edge. We use (\ref{K}), (\ref{E}), and
(\ref{M2}) to expand
\begin{eqnarray}
F &=& t + {\pd q \over \pd u_1} + {2(u_2 + \nu) \over (u_1 + \nu)} \ (t + {\pd q \over \pd u_2})
- {3(u_2 + \nu) \over 4 (u_1 + \nu)} \ (t + {\pd q \over \pd u_2}) \ s \nonumber \\
&& - \ {3(u_2 + \nu) \over 32(u_1 + \nu)} (t + {\pd q \over \pd u_2})
\ s^2 + O(s^3) \ , \label{F2}
\end{eqnarray}
and
\begin{eqnarray}
G &=& {\pi \over 2}  \left \{ [{u_2 + \nu \over 2(u_1 + \nu)} + {1 \over 16} \ (1 - {3(u_1 - u_2) \over
u_1 - u_3})](t + {\pd q \over \pd u_2}) - \right. \nonumber
\\ && \phantom{{\pi \over 2} } \left. -
 {(u_1 - u_3)(u_2 + \nu) \over u_1 + \nu} \
{\pd^2 q \over \pd u_2 \pd u_3} \right \}  s^2   
\nonumber \\ &&
 + {\pi \over 2} \left \{ {3 \over 128} (1 - 5 {u_1 - u_2 \over u_1 - u_3})
(t + {\pd q \over \pd u_2} ) -  
\right. \nonumber \\ &&  \phantom{{\pi \over 2}} \left.
- {(u_1 - u_3)(u_2 + \nu) \over 8(u_1 + \nu)} {\pd^2 q
\over \pd u_2 \pd u_3} \right \} s^3 + O(s^4) \ . \label{G2}
\end{eqnarray}
Taking the limits of $F =0$ and $G/s^2 = 0$ as $s \rightarrow 0$ and simplifying the
results a bit, we obtain the equations governing the trailing edge
\begin{equation}
U(u_1, u_3) := (u_1 + \nu)(t + {\pd q(u_1, u_3, u_3)\over \pd u_1}) + 2(u_3 + \nu) (t +
{\pd q(u_1, u_3, u_3) \over \pd u_2})  = 0\,,  \label{U}
\end{equation}
and
\begin{equation}
V(u_1, u_3) := [(u_3 + \nu) - {1 \over 4} (u_1 + \nu)]
(t + {\pd q \over \pd u_2}) - 2(u_1 - u_3)(u_3 + \nu)
\ {\pd^2 q \over \pd u_2 \pd u_3} = 0 \,. \label{V}
\end{equation}

Solving for $t$ from (\ref{U}) and substituting it into (\ref{V}), we use (\ref{q}) to simplify the result and
get
\begin{align}
\nonumber
W(u_1, u_3)  &:= [(u_3 + \nu) - {1 \over 4} (u_1 + \nu)] (u_1 + \nu) {\pd^2 q \over \pd u_1 \pd u_3} 
\\ \label{W} &\phantom{:=} \quad +
[(u_1 + \nu) + 2(u_3 + \nu)](u_3 + \nu) {\pd^2 q \over \pd u_2 \pd u_3}  = 0 \ .
\end{align}
Obviously, equations (\ref{U}) and (\ref{V}) are equivalent to equations (\ref{U})
and (\ref{W}).

We now solve equation (\ref{W}) for $u_3$ as a function of $u_1$ in the neighborhood
of $u_1=u_3=\hat{u}$. We use formula (\ref{qbeta}) and the symmetry of $q$ to write
\begin{eqnarray}
\label{int1}
{\pd^2 q(u_1, u_3, u_3) \over \pd u_1 \pd u_3} &=& {1 \over 16 \sqrt{2} \pi}
\int_{-1}^1 f''({1 + \mu \over 2} u_3 + {1 - \mu \over 2} u_1) {(1 - \mu^2) \over \sqrt{1 - \mu}}
\ d \mu \ , \\
{\pd^2 q(u_1, u_3, u_3) \over \pd u_2 \pd u_3} &=&
{1 \over 64 \sqrt{2} \pi}
\int_{-1}^1 f''({1 + \mu \over 2} u_3 + {1 - \mu \over 2} u_1) {(1 + \mu)^2 \over
\sqrt{1 - \mu}  }
\ d \mu  \ . \label{int2}
\end{eqnarray}

For $u_1 = \hat{u}$, it follows from (\ref{cf2}), (\ref{int1}) and (\ref{int2}) that equation
(\ref{W}) has only the solution $u_3 = \hat{u}$ in the neighborhood of $u_1=u_3
=\hat{u}$.

For $u_1$ which is a little bigger than $\hat{u}$, we will show that there is a unique
$u_3$ such that equation (\ref{W}) holds. By (\ref{cf2}), (\ref{int1}) and (\ref{int2}), we have
$$W(\hat{u}, \tilde{u}_3) > 0 \quad \mbox{for some $ \tilde{u}_3 < \hat{u}$} \ .$$
Hence, $$W(u_1, \tilde{u}_3) > 0  \quad \mbox{for $u_1$ a bit larger than $\hat{u}$} \ .
$$
For each of such $u_1$'s, we deduct from (\ref{cf2}) and (\ref{W}) again that
$$W(u_1, \hat{u}) < 0 \ .$$
By the mean value theorem, we show that, for each $u_1$ that is slightly larger than
$\hat{u}$, there exists a $u_3 < \hat{u}$ such that (\ref{W}) holds. It is easy to
check the uniqueness of $u_3$.

Therefore, (\ref{W}) determines $u_{3}$ as a function of $u_{1}$,
$u_{3}$ $=$ $A(u_{1})$, for small non-positive $u_{1}$ with
$A(\hat{u})$ $=$ $\hat{u}$.
The smoothness of the function $f(u)$ and Lemma \ref{q3} imply that $A(u_{1})$
is a smooth decreasing function of $u_{1}$.

Next, substituting $u_{3}$ $=$ $A(u_{1})$ into (\ref{U}), it is not hard to show that
(\ref{U}) determines $u_1$ as a function of $t$.  We have therefore
proved  the short time version of
the following lemma.

\begin{lem}
\label{x-} Under the conditions of Theorem \ref{main}, equations
(\ref{U}) and (\ref{V}) have a unique solution ($u_1^-(t)$,
$u_{2}^{-}(t)$, $u_{3}^{-}(t)$) with $u_{2}^{-}(t) = u_{3}^{-}(t)$
for all t $\geq$ 0. The solution has the property that
 $u_{1}^{-}(t)$ $>$
$u_{2}^{-}(t)$ = $u_{3}^{-}(t)$ for $t$ $>$ $0$ and that
 $u_{1}^{-}(0)$ =
$u_{2}^{-}(0) = u_{3}^{-}(0) = \hat{u}$.
\end{lem}
\begin{proof}
We will now extend the solution ($u_1^-(t)$, $u_{2}^{-}(t)$,
$u_{3}^{-}(t)$) of equations (\ref{U}) and (\ref{V}) for all $t>0$.
Before doing this, we need a lemma.
\begin{lem}
\label{trailing}
Under conditions of Theorem \ref{main}, the following hold:
\begin{eqnarray*}
\frac{\partial^{2} q}{\partial u_{1} \partial u_{2}}  =  \frac{
\partial^{2} q}{\partial u_{1} \partial u_{3}} < 0  \ , \quad
\frac{\partial^{2} q}{\partial u_{1}^{2}}  <  0 \ , \quad t + {\pd q \over \pd u_1}
< 0 \ , 
\end{eqnarray*}
and
\begin{equation*}
\quad t + {\pd q \over \pd u_2}=t + {\pd q \over \pd u_3} > 0
\, ,
\end{equation*}
at the solution $(u_{1}, u_{3},
u_{3})$ of (\ref{W}) where $u_{1} > u_{3}$.
\end{lem}
\begin{proof}
By (\ref{q}) and Lemma \ref{q3},
\begin{eqnarray}
\frac{\partial^{2} q}{\partial u_{1} \partial u_{3}} -
\frac{\partial^{2} q}{\partial u_{2} \partial u_{3}}  =  \frac{
\partial}{\partial u_{3}}[\frac{\partial q}{\partial u_{1}} -
\frac{\partial q}{\partial u_{2}}]
=  2(u_{1} - u_{2}) \frac{\partial^{3} q}{\partial u_{1}
\partial u_{2} \partial u_{3}}  < 0 \,, \label{132}
\end{eqnarray}
which when combined with (\ref{W}) gives
\begin{displaymath}
\frac{\partial^{2} q}{\partial u_{1} \partial u_{3}} < 0
\end{displaymath}
as long as $4(u_3 + \nu) - (u_1 + \nu) > 0$.  This inequality holds even when
$4(u_3 + \nu) - (u_1 + \nu) \leq 0$. To see this, suppose the inequality fails
at some point, at which
$\pd^2_{u_2 u_3} q$ must vanish because of (\ref{W}). This would violate (\ref{132}).

The other inequalities of Lemma \ref{trailing} are shown in the same
way.

\end{proof}

We now calculate the partial derivatives of $U$ and $V$ at the solution $(u_1, u_3, u_3)$
of (\ref{U}) and (\ref{V}),
\begin{eqnarray*}
{\pd U \over \pd u_1} &=& t + {\pd q \over \pd u_1} + (u_1 + \nu) {\pd^2 q \over
\pd u_1^2} + 2 (u_3 + \nu) {\pd^2 q \over \pd u_1 \pd u_3} < 0\,, \\
{\pd U \over \pd u_3} &=& 2(t + {\pd q \over \pd u_3}) + 2(u_1 + \nu)
{\pd^2 q \over \pd u_1 \pd u_3} + 8(u_3 + \nu) {\pd^2 q \over \pd u_2 \pd u_3} \\
&=& {1 \over u_1 - u_3} \left [ (u_1 + \nu) (t + {\pd q \over \pd u_1}) +
2(u_3 + \nu) (t + {\pd q \over \pd u_3}) \right ] = 0\,, \\
{\pd V \over \pd u_3} &=& t + {\pd q \over \pd u_2}
+ [(8(u_3 + \nu) - 3 (u_1 + \nu)] {\pd^2 q \over \pd u_2 \pd u_3} 
\\ && \phantom{t+} \quad -
2(u_1 - u_3)(u_3 + \nu) ({\pd^3 q \over \pd u_2^2 \pd u_3} +
{\pd^3 q \over \pd u_2 \pd u_3^2}) \\
&=& { 3(u_1 + \nu)^2 - 12 (u_1 + \nu)(u_3 + \nu) + 24 (u_3 + \nu)^2 \over
8(u_1 - u_3)(u_3 + \nu) } \ ( t + {\pd q \over \pd u_2}) \\
&& - \  2(u_1 - u_3)(u_3 + \nu)
({\pd^3 q \over \pd u_2^2 \pd u_3} + {\pd^3 q \over \pd u_2 \pd u_3^2}) > 0\,,
\end{eqnarray*}
where we have used Lemmas \ref{q3} and \ref{trailing} to determine 
the signs of the derivatives.
These show that the Jacobian
\begin{displaymath}
\frac{\partial (U, V)}{\partial (u_{1}, u_{3})} \neq 0
\end{displaymath}
on the solution $(u_{1},u_{3},u_{3})$ of equations (\ref{U}) and (\ref{V})
where $u_{1}
> u_{3}$.

Therefore, by the Implicit Function Theorem,  equations (\ref{U}) and (\ref{V})
can be solved for
$u_{1}^{-}(t)$, $u_{2}^{-}(t)$ =
$u_{3}^{-}(t)$ for all $t\geq 0$. Furthermore, it is easy to check that
$u_{1}^{-}(t)$ is an increasing function of $t$.
\end{proof}

\subsection{The leading edge}

\smallskip

At the leading edge, $u_1 = u_2$, i.e., $s=1$, it follows from
(\ref{K2}), (\ref{E2}), (\ref{F}), and (\ref{G}) that equations (\ref{FG}) turn out
to be
$$t + {\pd q \over \pd u_1} (u_1, u_1, u_3) = 0 \ , \quad
t + {\pd q \over \pd u_3} (u_1, u_1, u_3) = 0  \ .$$ In the same way
as we handle Lemma \ref{x-}, we can solve the above equations
for $u_1$ and $u_3$ as functions of $t$, leading to the lemma.
\begin{lem}
\label{leading} Under the conditions of Theorem \ref{main}, system
(\ref{FG}) has a unique solution ($u_{1}^{+}(t)$, $u_{2}^{+}(t)$,
$u_{3}^{+}(t)$) with $u_{1}^{+}(t) = u_{2}^{+}(t)$ for all $t \geq
0$. The solution has the property that $u_{1}^{+}(t) = u_{2}^{+}(t)
> u_{3}^{+}(t)$ for $t > 0$ and that $u_{1}^{+}(0) = u_{2}^{+}(0) =
u_{3}^{+}(0) = \hat{u}$.
\end{lem}

\subsection{Near the trailing edge}

\smallskip

By Lemma \ref{x-}, ($u_{1}^{-}(t)$, $u_{2}^{-}(t)$, $u_{3}^{-}(t)$)
satisfies equations (\ref{U}) and (\ref{V}). For each fixed t $>$ 0, we need to
solve equations (\ref{FG}) for
$u_{1}$
and $u_{3}$ as functions of $u_{2}$ in the neighborhood of
$u_{2}^{-}(t)$. This is carried out in
\begin{lem}
\label{two}
For each t $>$ 0, equations (\ref{FG}) can be solved for
$u_{1}$ and $u_{3}$
in terms of $u_{2}$ in the neighborhood of ($u_{1}^{-}(t)$,
$u_{2}^{-}(t)$, $u_{3}^{-}(t)$)
\begin{equation}
\label{MNf}
\left\{ \begin{array}{ll}
          u_{1} = M ( u_{2} )   \\
          u_{3} = N ( u_{2} )
          \end{array}
\right.
\end{equation}
such that $u_{1}^{-}(t) = M(u_{2}^{-}(t))$ and $u_{3}^{-}(t)
= N(u_{2}^{-}(t))$, Moreover, for $u_2 > u_{2}^{-}(t)$,
\begin{equation}
\label{MN} N(u_{2}) < u_{2} < M(u_{2})\,.
\end{equation}
\end{lem}
\begin{proof}
Calculating the first partial derivatives of $F$ and $G/s^2$ of (\ref{F2}) and (\ref{G2}) at
($u_{1}^{-}(t)$, $u_{2}^{-}(t)$,
$u_{3}^{-}(t)$), where $u_2^{-}(t) = u_3^{-}(t)$,  and using (\ref{q}), we find
\begin{eqnarray*}
{\pd F \over \pd u_1} &=& {\pd^2 q \over \pd u_1^2} - {2(u_2 + \nu) \over
(u_1 + \nu)^2 } (t + {\pd q \over \pd u_2}) + {2(u_2 + \nu) \over u_1 + \nu}
{\pd^2 q \over \pd u_1 \pd u_2} < 0 \ , \\
{\pd F \over \pd u_2} &=& {\pd^2 q \over \pd u_1 \pd u_2} + [ {2 \over u_1 + \nu}
- {3 \over 4} {1 \over u_1 - u_3} ] (t + {\pd q \over \pd u_2}) +
\\ && \phantom{{\pd^2 q \over \pd u_1^2}}
+ {6(u_2 + \nu)
\over u_1 + \nu} {\pd^2 q \over \pd u_2 \pd u_3} = 0 \ ,  \\
{\pd F \over \pd u_3} &=& {\pd^2 q \over \pd u_1 \pd u_3} + {3 \over 4} {1 \over
u_1 - u_3} (t + {\pd q \over \pd u_2}) + {2(u_2 + \nu) \over u_1 + \nu}
{\pd^2 q \over \pd u_2 \pd u_3} = 0 \ , \\
{\pd (G/s^2) \over \pd u_2} &=& {\pi \over 2} \left \{ {19(u_1 + \nu) - 16(u_3 + \nu)
\over 32(u_1 - u_3)(u_1 + \nu)}
(t + {\pd q \over \pd u_2}) - {u_1 - u_2 \over 2(u_1 + \nu) }
{\pd^2 q \over \pd u_2 \pd u_3} \right \} \\
&=& {  2(u_1 + \nu)^2 + 9(u_1 + \nu)(u_3 + \nu) - 8(u_3 + \nu)^2  \over 64
(u_1 - u_3)(u_1 + \nu) (u_2 + \nu) } \ \pi (t + {\pd q \over \pd u_2}) > 0 \ , \\
{\pd (G/s^2) \over \pd u_3} &=& {\pi \over 2} \left\{ [- {3 \over 16(u_1 - u_3)}
+{3(u_1 + \nu) \over 32(u_1 - u_3)(u_2 + \nu)} ] (t + {\pd q \over \pd u_2}) + \right. \\ && \phantom{{\pi \over 2}} \left.
+ {3(u_2 + \nu) \over 2(u_1 + \nu) } {\pd^2 q \over \pd u_2 \pd u_3}
- {(u_1 - u_3)(u_2 + \nu) \over u_1 + \nu}{\pd^3 q \over \pd u_2 \pd u_3^2} \right\} \\
&=& {\pi \over 2} \left\{ {3 [(u_1 + \nu)^2 - 4(u_1 + \nu)(u_3 + \nu) +
8(u_3 + \nu)^2] \over 32 (u_1 - u_3) (u_1 + \nu) (u_2 + \nu) } \
(t + {\pd q \over \pd u_2}) \right.
\\ & & \left.  \phantom{\pi/2 \{ }
-  {(u_1 - u_3)(u_2 + \nu) \over u_1 + \nu}
{\pd^3 q \over \pd u_2 \pd u_3^2} \right \} > 0 \ ,
\end{eqnarray*}
where we have used (\ref{U}) and (\ref{V}) to simplify the results, and
Lemmas \ref{q3} and  \ref{trailing} to determine the signs of the derivatives.

These prove the non-vanishing of the Jacobian
\begin{displaymath}
\frac{\partial (F, G/s^2)}{\partial (u_{1}, u_{3})}
\end{displaymath}
at ($u_{1}^{-}(t)$, $u_{2}^{-}(t)$, $u_{3}^{-}(t)$). Hence,
equations (\ref{FG}) can be solved for
\begin{equation*}
\left\{ \begin{matrix}
u_{1} = M(u_2) \\
u_3 = N(u_2)
\end{matrix} \right.
\end{equation*}
in a neighborhood of $u_{2}^{-}(t)$ such that $u_{1}^{-}(t)$ =
$M(u_{2}^{-}(t))$ and $u_{3}^{-}(t)$ = $N(u_{2}^{-}(t))$.
Furthermore, $N(u_2)$ is a decreasing function of $u_2$ and so (\ref{MN}) holds.
\end{proof}

\subsection{The passage from the trailing edge to the leading edge}

\smallskip

We shall show that, for each fixed t $>$ 0, solutions (\ref{MNf}) of
equations (\ref{FG})
can be further extended as long as $N (u_{2})$ $<$ $u_{2}$
$<$ $M (u_{2})$. The Jacobian of system (\ref{1}) and (\ref{2}) with respect
to $(u_{1}, u_{3})$ has to be estimated along the extension.
\begin{lem}
\label{im}
Under the conditions of Theorem \ref{main}, the following
inequalities hold for each $t > 0$.
\begin{equation}
\label{important}
\frac{\partial (\lambda_{1}t + w_{1})}{\partial u_{1}} < 0,  \quad
\frac{\partial (\lambda_{2}t + w_{2})}{\partial u_{2}} > 0,  \quad
\frac{\partial (\lambda_{3}t + w_{3})}{\partial u_{3}} < 0
\end{equation}
on the solution ($u_{1}$, $u_{2}$, $u_{3}$) of (\ref{1}) and (\ref{2})
(or equivalently (\ref{FG})) in the region
$0 < u_{3} + \nu $ $<$ $u_{2} + \nu $ $<$ $u_{1} + \nu$.
\end{lem}
\begin{proof}
Using formulae (\ref{w}) for $w_1$, $w_2$, and $w_3$, we see that (\ref{1}) and (\ref{2}) are equivalent to
\begin{equation}
\label{12}
[\lambda_{1} - (u_{1} + u_{2} + u_{3} + 2 \nu)](t + \frac{
\partial q}{\partial u_{1}}) = [\lambda_{2} - (u_{1} + u_{2}
+ u_{3} + 2 \nu)](t + \frac{\partial q}{\partial u_{2}}) \,,
\end{equation}
\begin{equation}
\label{23}
[\lambda_{2} - 2(u_{1} + u_{2} + u_{3} + 2 \nu)](t + \frac{
\partial q}{\partial u_{2}}) = [\lambda_{3} - (u_{1} + u_{2}
+ u_{3} + 2 \nu)](t + \frac{\partial q}{\partial u_{3}})  \,.
\end{equation}
By Lemma \ref{trailing},
\begin{equation}
\label{1213}
\frac{\partial^{2} q(u_1, u_2, u_3)}{\partial u_{1} \partial u_{2}} < 0, \quad
\frac{\partial^{2} q(u_1, u_2, u_3)}{\partial u_{1} \partial u_{3}} < 0
\end{equation}
at the trailing edge.

We claim that inequalities (\ref{1213}) hold on the solution $(u_1, u_2, u_3)$
of (\ref{1}) and (\ref{2})
as long as $0< u_{3} + \nu $
$<$ $u_{2} + \nu $ $<$ $u_{1} + \nu$.

We justify the claim by contradiction. Suppose otherwise, for
instance at some point $(\bar{u}_1, \bar{u}_2, \bar{u}_3)$ on the
solution of (\ref{1}) and (\ref{2}), with $0 < \bar{u}_3 + \nu <
\bar{u}_2 + \nu < \bar{u}_1 + \nu$,
\begin{displaymath}
\frac{\partial^{2} q}{\partial u_{1} \partial u_{2}} = 0  \ .
\end{displaymath}
In view of (\ref{q}), this gives
\begin{displaymath}
\frac{\partial q}{\partial u_{1}} = \frac{\partial q}{\partial u_{2}}
\quad \mbox{at  $(\bar{u}_{1}, \bar{u}_{2}, \bar{u}_{3})$} \ ,
\end{displaymath}
which together with (\ref{l1>}), (\ref{l2<}), and (\ref{12}) imply
\begin{equation}
\label{1=2=0}
t + \frac{1}{2}\frac{\partial q}{\partial u_{1}} =
t + \frac{1}{2}\frac{\partial q}{\partial u_{2}} = 0 \ .
\end{equation}
By (\ref{l2<}), (\ref{l3<}), (\ref{23}), and (\ref{1=2=0}), we
obtain
\begin{displaymath}
t + \frac{\partial q}{\partial u_{3}} = 0\, ,
\end{displaymath}
which together with (\ref{q}) gives
\begin{equation}
\label{=0}
\frac{\partial^{2} q}{\partial u_{1} \partial u_{2}} =
\frac{\partial^{2} q}{\partial u_{1} \partial u_{3}} = 0
\end{equation}
at $(\bar{u}_{1}, \bar{u}_{2}, \bar{u}_{3})$.

On the other hand, by (\ref{q}) and Lemma \ref{q3},
\begin{equation*}
\frac{\partial^{2} q}{\partial u_{1} \partial u_{2}} -
\frac{\partial^{2} q}{\partial u_{1} \partial u_{3}} =
2(u_{2} - u_{3}) \frac{\partial^{3} q}{\partial u_{1} \partial
u_{2} \partial u_{3}} < 0
\end{equation*}
at $(\bar{u}_{1}, \bar{u}_{2}, \bar{u}_{3})$. This contradicts
(\ref{=0}) and the claim has been justified.

By (\ref{q}), we have
\begin{displaymath}
2(u_{1} - u_{3}) \frac{\partial^{2} q}{\partial u_{1} \partial
u_{3}} = \frac{\partial q}{\partial u_{1}} - \frac{\partial q}
{\partial u_{3}} \,.
\end{displaymath}
Differentiating this with respect to $u_{1}$ yields
\begin{eqnarray}
\frac{\partial^{2} q}{\partial u_{1}^{2}}  =  3\frac{\partial^{2} q}
{\partial u_{1}
\partial u_{3}} + 2(u_{1} - u_{3})\frac{\partial^{3} q}{
\partial u_{1}^{2} \partial u_{3}}  < 0  \ , \label{11}
\end{eqnarray}
where we have used (\ref{1213}) and Lemma \ref{q3} in the last step.

It follows from (\ref{1213}) and (\ref{q}) that
\begin{equation*}
\frac{\partial q}{\partial u_{1}} < \frac{\partial q}{\partial
u_{2}} \ , \quad \frac{\partial q}{\partial u_{1}} < \frac{\partial
q}{\partial u_{3}}\,,
\end{equation*}
which when combined with (\ref{l1>}), (\ref{l2<}), (\ref{l3<}), (\ref{12}) and
(\ref{23}) gives
\begin{equation}
\label{ineq}
t + \frac{1}{2}\frac{\partial q}{\partial u_{1}} < 0, \quad
t + \frac{1}{2}\frac{\partial q}{\partial u_{2}} > 0, \quad
t + \frac{1}{2}\frac{\partial q}{\partial u_{3}} > 0
\end{equation}
on the solution ($u_{1}$, $u_{2}$, $u_{3}$) of (\ref{1}) and (\ref{2})
in the region
$0 < u_{3} + \nu < u_{2} + \nu <u_{1} + \nu$.

Therefore, by (\ref{w}),
\begin{eqnarray*}
\frac{\partial (\lambda_{1}t + w_{1})}{\partial u_{1}} & = & \frac{\partial
\lambda_{1}}{\partial u_{1}}(t +
\frac{\partial q}{\partial u_{1}}) + [\lambda_{1}
- (u_{1} + u_{2} + u_{3} + 2 \nu)] \frac{\partial^{2} q}{\partial
u_{1}^{2}}  < 0 \  ,
\end{eqnarray*}
where in the last inequality we have used (\ref{l1>}), (\ref{11}), (\ref{ineq}), and
$$\frac{\partial \lambda_{1}}{\partial u_{1}} = {I \pd^2_{u_i u_i} I \over (\pd_{ui} I )^2}
> 0 \ . $$
This proves the first inequality of (\ref{important}).

Next we shall prove the rest of Lemma \ref{im}. By (\ref{q}), we have
\begin{displaymath}
2(u_{2} - u_{3})\frac{\partial^{2} q}{\partial u_{2} \partial
u_{3}} = \frac{\partial q}{\partial u_{2}} - \frac{\partial
q}{\partial u_{3}}\,.
\end{displaymath}
Differentiating this with respect to $u_{2}$ yields
\begin{equation}
\label{22}
\frac{\partial^{2} q}{\partial u_{2}^{2}} = 3\frac{\partial^{2} q}{
\partial u_{2} \partial u_{3}} + 2(u_{2} - u_{3})\frac{
\partial^{3} q}{\partial u_{2}^{2} \partial u_{3}}\,.
\end{equation}

Using (\ref{q}) to rewrite (\ref{23}), we obtain
\begin{equation}
\label{sign}
(\lambda_{2} - \lambda_{3})[t + \frac{\partial q}{\partial
u_{3}}] + 2 [\lambda_{3} - (u_{1} + u_{2} + u_{3} + 2 \nu )]\frac{
\partial^{2} q}{\partial u_{2} \partial u_{3}}(u_{2} - u_{3})
= 0
\end{equation}
which together with (\ref{22}) gives
\begin{eqnarray}
3\frac{\lambda_{2} - \lambda_{3}}{u_{2} - u_{3}}(t +
\frac{\partial q}{\partial u_{3}}) + 2 [\lambda_{2} - (u_{1} +
u_{2} + u_{3} + 2 \nu)]\frac{\partial^{2} q}{\partial u_{2}^{2}} \nonumber  \\
 =  4[\lambda_{2} - (u_{1} + u_{2} + u_{3} + 2 \nu)]
(u_{2} - u_{3})
\frac{\partial^{3} q}{\partial u_{2}^{2} \partial u_{3}}  > 0  \ , \label{2323}
\end{eqnarray}
where we have used (\ref{l3<}) and Lemma \ref{q3} in the last inequality.

It follows from (\ref{w}) that
\begin{eqnarray*}
\frac{\partial (\lambda_{2}t + w_{2})}{\partial u_{2}} & = & \frac{
\partial \lambda_{2}}{\partial u_{2}}(t + \frac{\partial q}{
\partial u_{2}}) + [\lambda_{2} - (u_{1} + u_{2} +
u_{3} + 2 \nu )]\frac{
\partial^{2} q}{\partial u_{2}^{2}} \\
& > & \frac{3}{2}\frac{\lambda_{2} - \lambda_{3}}{u_{2} - u_{3}}
(t + \frac{\partial q}{
\partial u_{2}}) + [\lambda_{2} - (u_{1} + u_{2} +
u_{3} + 2 \nu )]\frac{
\partial^{2} q}{\partial u_{2}^{2}} \\
&=& \frac{3}{2}\frac{\lambda_{2} - \lambda_{3}}{u_{2} - u_{3}}
(t + \frac{\partial q}{
\partial u_{3}}) + [\lambda_{2} - (u_{1} + u_{2} +
u_{3} + 2 \nu )]\frac{
\partial^{2} q}{\partial u_{2}^{2}} 
\\ && \phantom{\frac{3}{2}} +
 3(\lambda_{2} - \lambda_{3}) \frac{
\partial^{2} q}{\partial u_{2} \pd u_3} \\
&>& 3(\lambda_{2} - \lambda_{3}) \frac{
\partial^{2} q}{\partial u_{2} \pd u_3} \geq 0  \ ,
\end{eqnarray*}
where we have used Lemma \ref{IM}, (\ref{ineq}) in the first inequality,
and (\ref{2323})
in the second one. The last inequality is due to the fact that
$\lambda_{2} - \lambda_{3}$ and $\pd^2_{u_2 u_3} q$ have the same sign because of
(\ref{sign}).

This proves the second inequality of (\ref{important}). In the same
way, we can prove the last one.
\end{proof}

We are ready to conclude the proof of Theorem \ref{main}.

Proof of Theorem \ref{main}: By Lemma \ref{two}, equations
(\ref{FG}) can be solved for
\begin{displaymath}
\left \{\begin{array}{ll}
          u_{1} = M(u_{2}) \\
          u_{3} = N(u_{2})
        \end{array}
\right.
\end{displaymath}
in the neighborhood of ($u_{1}^{-}(t)$, $u_{2}^{-}(t)$, $u_{3}^{-}
(t)$). Furthermore, (\ref{MN}) holds if $u_{2}$ $>$ $u_{2}^{-}(t)$. We shall
extend the solution in the positive $u_{2}$ direction as far
as possible.
It follows from Corollary \ref{c} and Lemma \ref{im} that, along the extension of (\ref{MNf})
in the region $u_{1} + \nu > u_{2} + \nu > u_{3} + \nu > 0$,
the Jacobian matrix of (\ref{1}) and (\ref{2})
is diagonal and therefore is nonsingular. Furthermore, equations (\ref{1})
and (\ref{2})
determines (\ref{MNf}) as two decreasing functions of $u_{2}$.

This immediately guarantees that (\ref{MNf}) can be extended as far
as necessary in the region $u_{1} + \nu > u_{2} + \nu > u_{3} + \nu
> 0$. Since $M(u_{2})$ is decreasing, (\ref{MNf}) stops at some
point $u_{2}^{+}(t)$ where, obviously, $M(u_{2}^{+}(t))$ =
$u_{2}^{+} (t)$. Therefore, we have shown that (\ref{1}) and
(\ref{2}) determine $u_{1}$ and $u_{3}$ as decreasing functions of
$u_{2}$ over the interval $[u_{2}^{-}(t), u_{2}^{+}(t)]$.

Let
\begin{equation*}
\left\{ \begin{matrix}
u_{1}^{+}(t)  =  M (u_{2}^{+}(t)) \\
u_{3}^{+}(t)   =  N (u_{2}^{+}(t))\,.
\end{matrix} \right. \,
\end{equation*}
Clearly, $(u_{1}^{+}(t), u_{2}^{+}(t), u_{3}^{+}(t))$ solves
system (\ref{FG}) at the leading edge $u_{1} = u_{2}$. Hence,
these $u_{1}^{+}(t)$, $u_{2}^{+}(t)$ and $u_{3}^{+}(t)$ are
exactly the ones appearing in Lemma \ref{leading}.

Substituting (\ref{MNf}) into (\ref{HO}), we obtain
\begin{displaymath}
x = \lambda_{2}(M(u_{2}), u_{2}, N(u_{2})) t + w_{2}(M(u_{2}),
u_{2}, N(u_{2}))
\end{displaymath}
which by Corollary \ref{c} and Lemma \ref{im} clearly determines $x$ as an
increasing function of $u_2$ over interval
$[u_2^{-}(t), u_{2}^{+}(t)]$. It
follows that, for each fixed $t > 0$, $u_{2}$ is a function of $x$ over
the interval [$x^{-}(t)$, $x^{+}(t)$], and that therefore so are $u_{1}$
and $u_{3}$, where
\begin{eqnarray}
\label{pm}
x^{\pm}(t) = \lambda_{2}(u_{1}^{\pm}(t), u_{2
}^{
\pm}(t), u_{3}^{\pm}(t)) t +
w_{2}(u_{1}^{\pm}(t), u_{2}^{
\pm}(t), u_{2}^{\pm}(t))  \,.
\end{eqnarray}
Thus, (\ref{HO}) can be solved for
\begin{eqnarray*}
u_{1} = u_{1}(x, t) \ , \quad
u_{2} = u_{2}(x, t) \ , \quad
u_{3} = u_{3}(x, t)
\end{eqnarray*}
within a wedge
\begin{equation}
\label{wedge}
\begin{array}{ll}
x^{-}(t) < x < x^{+}(t) \hspace*{.5in} for ~ t > 0 \ , \\
x^{-}(0) = x^{+}(0) = 0  \ ,
\end{array}
\end{equation}
where we have used (\ref{pm}), Lemma \ref{x-}, and Lemma
\ref{leading} in the last equations.

Boundary conditions (\ref{bc1}) and (\ref{bc2}) can be checked easily.
The proof of Theorem \ref{main} would be completed if we can verify that
the wedge is indeed a cusp. First we need a lemma.
\begin{lem}
\label{cusp} At $(u_{1}^{-}(t)$, $u_{2}^{-}(t)$, $u_{3}^{-}(t))$,
\begin{equation}
\label{tr2}
\frac{\partial (\lambda_{2}t + w_{2})}{\partial u_{2}} =
\frac{\partial (\lambda_{3}t + w_{3})}{\partial u_{3}} = 0\,,
\end{equation}
while at $(u_{1}^{+}(t)$, $u_{2}^{+}(t)$, $u_{3}^{+}(t))$,
\begin{equation}
\label{le2}
\frac{\partial (\lambda_{1}t + w_{1})}{\partial u_{1}} =
\frac{\partial (\lambda_{2}t + w_{2})}{\partial u_{2}} = 0\,.
\end{equation}
\end{lem}
\begin{proof}

Using expansions (\ref{K}) and (\ref{E}), we obtain from (\ref{lambda2u2}) that
\begin{displaymath}
\frac{\partial \lambda_{2}}{\partial u_{2}} =
{6[(u_3 + \nu) - {1 \over 4} (u_1 + \nu)] \over u_1 + \nu}
\end{displaymath}
at ($u_{1}^{-}(t)$, $u_{2}^{-}(t)$, $u_{3}^{-}(t)$).

By (\ref{w}), we have
\begin{eqnarray*}
\frac{\partial (\lambda_{2}t + w_{2})}{\partial u_{2}}  &=&
 {6[(u_3 + \nu) - {1 \over 4} (u_1 + \nu)] \over (u_1 + \nu)}
(t + \frac{\partial q}{\partial u_{2}}) 
\\ && 
\ - {4 (u_1 - u_3)(u_3 + \nu) \over u_1 + \nu}
\ \frac{\partial^{2} q}{\partial u_{2}^{2}}
 =  0 \ ,
\end{eqnarray*}
where we have used (\ref{tr}) in the first equality, and 
(\ref{V}) and (\ref{22}) in the last equality. The proof
also applies to
the other equation of (\ref{tr2}).

To prove (\ref{le2}), we proceed as follows. By (\ref{lambda}),
(\ref{K2}), and (\ref{E2}), we find
\begin{equation}
\label{1-s}
\frac{\partial \lambda_{2}}{\partial u_{1}} = O (K(s)) \quad \mbox{as $s \rightarrow
1 $ \ .}
\end{equation}
By Corollary \ref{c} and formula (\ref{w}), we calculate the derivative
on the solution of
(\ref{1}) and (\ref{2})
\begin{eqnarray*}
0 & = & \frac{\partial (\lambda_{2}t + w_{2})}{\partial u_{1}} \\
  & = & \frac{\partial \lambda_{2}}{\partial u_{1}}(t +
\frac{\partial q}{\partial u_{2}}) + \frac{1}{2}[\lambda_{2} - (u_{1}
+ u_{2} + u_{3} + 2 \nu)]\frac{\partial^{2} q}{\partial u_{1} \partial
u_{2}} + \frac{\partial q}{\partial u_{1}} - \frac{\partial q}
{\partial u_{2}} \,,
\end{eqnarray*}
which when combined with (\ref{le}) and (\ref{1-s}) gives
\begin{equation}
\label{1-s2}
\lim_{s \rightarrow 1} [t + \frac{\partial q}{\partial
u_{2}} ] \ K(s) = 0 \,.
\end{equation}
On the other hand, as in (\ref{1-s}), we have
\begin{equation}
\label{1-s3}
\frac{\partial \lambda_{2}}{\partial u_{2}} = O (K(s)) \quad  \mbox{as $ s
\rightarrow 1$ \ .}
\end{equation}
Therefore, by (\ref{w}) we see that at $(u_{1}^{+}$, $u_{2}^{+}$,
$u_{3}^{+})$
\begin{eqnarray*}
\frac{\partial (\lambda_{2}t + w_{2})}{\partial u_{2}}  =  \frac{
\partial \lambda_{2}}{\partial u_{2}}(t + \frac{1}{2}\frac{\partial q}{
\partial u_{2}}) + [\lambda_{2} - (u_{1} + u_{2} +
u_{3} + 2 \nu)]\frac{
\partial^{2} q}{\partial u_{2}^{2}}
 =  0 \ ,
\end{eqnarray*}
where, by (\ref{1-s2}) and (\ref{1-s3}), the first term vanishes, while the second term vanishes
because of (\ref{le}). This proves the second equality of (\ref{le2}). In the same way, we
can check the other equality of (\ref{le2}).
\end{proof}

Now we continue to finish the proof of Theorem \ref{main}.
Differentiating (\ref{pm}) with respect to $t$, by Corollary \ref{c}
and Lemma \ref{cusp} we obtain
\begin{eqnarray*}
\frac{d x^{\underline{+}}(t)}{d t} & = & \frac{\partial (\lambda_{2}t + w_{2})}
{\partial u_{2}}\frac{d u_{2}^{\pm}}{d t} + \lambda_{2}(
u_{1}^{\pm}, u_{2}^{\pm}, u_{3}^{\pm
})   \\
  & = & \lambda_{2}(u_{1}^{\pm}, u_{2}^{\pm},
u_{3}^{\pm})\,,
\end{eqnarray*}
which when combined with (\ref{tr}), (\ref{le}), Lemma \ref{x-}, and
Lemma \ref{leading} gives
\begin{displaymath}
\frac{d x^{\pm}(t)}{d t} = 3 \hat{u} + 2 \nu \quad \mbox{at  $ t = 0 $} \ .
\end{displaymath}
Therefore, wedge (\ref{wedge}) is a cusp. This completes the proof of
Theorem \ref{main}.

We immediately conclude from Theorem \ref{th1} and Theorem \ref{main} the following result on the
initial value problem of the Whitham equations.
\begin{thm}
\label{Main}
For a decreasing initial function $u_0(x)$ whose inverse function $f(u)$ satisfying the conditions of Theorem
\ref{main}, the Whitham equations (\ref{CHW}) have a solution $(u_1(x,t), u_2(x,t), u_3(x,t))$ within
a cusp for all positive time. The Burgers solution of (\ref{we0}) exists 
outside the cusp.
The Whitham solution matches the Burgers solution
on the boundary of the cusp in the fashion of (\ref{bc1}) and (\ref{bc2}).
\end{thm}

We close this paper with two observations. First, it is obvious from
the proof of Theorem \ref{Main} that one should expect local (in
time) results if local conditions are assumed. Namely, if the global
condition $f'''(u) < 0$ for all $b < u < a$ in Theorem \ref{Main} is
replaced by a local condition $f'''(u) < 0$ in the neighborhood of
the breaking point $\hat{u}$, the results of Theorem \ref{Main} are
only true for a short time after the breaking time.

Second, a hump-like initial function
can be decomposed into a decreasing and an increasing
parts. It is known that the decreasing part causes the Burgers solution of (\ref{we0})
to develop finite time singularities
while the increasing part does not. These two pieces
of data would not interact with each other for a short time after the breaking of
the Burgers solution. As a consequence, a short time result also holds for a
hump-like initial function.

\appendix
\section{Complete Elliptic Integrals}

In this Appendix, we list some of the well-known properties of the complete elliptic
integrals of the first, second and third kind.

These are the derivative formulae
\begin{eqnarray}
\frac{d K(s)}{d s} & = & \frac{E(s) - (1-s)K(s)}{2s(1-s)} \ , \label{K3} \\
\frac{d E(s)}{d s} & = & \frac{E(s) - K(s)}{2s} \ , \label{E3} \\
\frac{d \Pi(\rho,s)}{d \rho} & = & \frac{\rho E(s) + (s - \rho) K(s) + (\rho^2- s) \Pi(\rho, s) }{2\rho (1 - \rho) (\rho -s)} \ , \label{P3} \\
\frac{d \Pi(\rho, s)}{d s} & = & \frac{E(s) - (1-s)\Pi(\rho, s)}{2(1-s)(s-\rho)} \ . \label{P3'}
\end{eqnarray}

$K(s)$ and $E(s)$ have the
expansions
\begin{eqnarray}
K(s) & = & \frac{\pi}{2} [1 + \frac{s}{4} + \frac{9}{64} s^{2}
+ \cdots + (\frac{1 \cdot 3 \cdots (2n-1)}{2 \cdot 4 \cdots 2n})^{2} s^{n}
+ \cdots] \ , \label{K}\\
E(s) & = & \frac{\pi}{2} [1 - \frac{s}{4} - \frac{3}{64} s^{2}
 - \cdots - \frac{1}{2n-1}(\frac{1 \cdot 3 \cdots (2n-1)}{
2 \cdot 4 \cdots 2n})^{2} s^{n} -
\cdots]  \label{E}
\end{eqnarray}
for $|s| < 1$. They also have the asymptotics
\begin{eqnarray}
K(s) & \approx & \frac{1}{2} \log \frac{16}{1 - s}  \label{K2} \ , \\
E(s) & \approx & 1 + \frac{1}{4}(1 - s)[\log \frac{16}{1 - s} - 1] \label{E2}
\end{eqnarray}
when $s$ is close to $1$. They satisfy the inequalities \cite{Tian1}
\begin{equation}
\frac{1}{1 - \frac{s}{2}} < \frac{K(s)}{E(s)} < \frac{1-\frac{s}{2}}{1-s}
\hspace*{.5in} for ~ 0 < s < 1 \ . \label{KE}
\end{equation}

The complete elliptic integral of the third kind has the following
behavior
\begin{eqnarray}
\Pi(\rho, s) &=& {\pi \over 2} \quad \mbox{when $\rho = 0$, $s = 0$} \ , \label{P1} \\
{\Pi(\rho, s) \over K(s)} &\approx& {1 \over 1 - \rho} \quad
\mbox{when $s$ is close to $1$} \ . \label{P2}
\end{eqnarray}

\bigskip

{\bf Acknowledgments.}
T.G. was supported in part by the MISGAM program
of the European Science Foundation, and by the RTN ENIGMA and Italian
COFIN 2006 ``Geometric methods in the theory of nonlinear waves and their
applications.''
V.P. was supported in part by NSF Grant DMS-0135308.
F.-R. T. was supported in part by
NSF Grant DMS-0404931.

\bibliographystyle{amsplain}

\end{document}